\documentclass[journal,12pt,draftcls,onecolumn]{IEEEtran}
\usepackage{multirow,hhline}

\usepackage{graphics,graphicx}
\usepackage[cmex10]{amsmath}
\usepackage{amssymb}
\usepackage{algorithm,algorithmic}
\usepackage{array,subeqnarray}
\usepackage[tight,footnotesize]{subfigure}
\usepackage{epsfig,psfrag}
\usepackage{hyperref}
\usepackage{cite}
\usepackage{import,color}
\usepackage{epstopdf}
\usepackage{array}
\newcolumntype{L}[1]{>{\raggedright\let\newline\\\arraybackslash\hspace{0pt}}m{#1}}
\newcolumntype{C}[1]{>{\centering\let\newline\\\arraybackslash\hspace{0pt}}m{#1}}
\newcolumntype{R}[1]{>{\raggedleft\let\newline\\\arraybackslash\hspace{0pt}}m{#1}}
\newcommand{\minitab}[2][l]{\begin{tabular}{#1}#2\end{tabular}}

\newtheorem{theorem}{Theorem}[section]

\newtheorem{lemma}[theorem]{Lemma}
\newtheorem{remark}[theorem]{Remark}

\newtheorem{proof}[theorem]{Proof}

\def\bbeta{{\boldsymbol{\beta}}}

\def\be{{\boldsymbol{e}}}
\def\bf{{\boldsymbol{f}}}

\def\bh{{\boldsymbol{h}}}

\def\bw{{\boldsymbol{w}}}
\def\bx{{\boldsymbol{x}}}

\def\bH{{\boldsymbol{H}}}

\def\bR{{\boldsymbol{R}}}

\def\bbC{{\mathbb{C}}}

\def\bbE{{\mathbb{E}}}

\def\calS{{{\mathcal{S}}}}

\def\calCN{{{\mathcal{CN}}}}

\def\calK{{{\mathcal{K}}}}

\def\calC{{{\mathcal{C}}}}

\def\trans{{\textsf{T}}}
\def\hermit{{\textsf{H}}}

\def\t{{\textsf{t}}}

\def\az{{\textsf{az}}}
\def\el{{\textsf{el}}}

\newlength{\figwidth}
\newlength{\figwidthb}
\title{Adaptive Multicell 3D Beamforming in Multi-Antenna Cellular Networks\thanks{N. Seifi is with Ericsson Research, SE-164 80 Stockholm, Sweden (e-mail: nima.seifi@ericsson.com).}\thanks{R. W. Heath Jr. is with the Wireless Networking and Communications Group (WNCG), Department of Electrical and Computer Engineering, The University of Texas at Austin, Austin, TX 78712-0240 USA (e-mail: rheath@ece.utexas.edu).}\thanks{M. Coldrey is with Ericsson Research, SE-417 56 Gothenburg, Sweden (e-mail: mikael.coldrey@ericsson.com).}\thanks{T. Svensson is with the Department of Signals and Systems, Chalmers University of Technology, SE-412 96 Gothenburg, Sweden (e-mail: tommy.svensson@chalmers.se).}}

\author{Nima~Seifi, Robert~W.~Heath~Jr., Mikael~Coldrey, and Tommy~Svensson}

\begin{document}

\maketitle
%\vspace{-50pt}
\begin{abstract}
We consider a cellular network with multi-antenna base stations (BSs) and single-antenna users, multicell cooperation, imperfect channel state information, and directional antennas each with a vertically adjustable beam. We investigate the impact of the elevation angle of the BS antenna pattern, denoted as tilt, on the performance of the considered network when employing either a conventional single-cell transmission or a fully cooperative multicell transmission. Using the results of this investigation, we propose a novel hybrid multicell cooperation technique in which the intercell interference is controlled via either \emph{cooperative beamforming} in the horizontal plane or \emph{coordinated beamfroming} in the vertical plane of the wireless channel, denoted as \emph{adaptive multicell 3D beamforming}. The main idea is to divide the coverage area into two disjoint \emph{vertical regions} and adapt the multicell cooperation strategy at the BSs when serving each region. A fair scheduler is used to share the time-slots between the vertical regions. It is shown that the proposed technique can achieve performance comparable to that of a fully cooperative transmission but with a significantly lower complexity and signaling requirements. To make the performance analysis computationally efficient, analytical expressions for the user ergodic rates under different beamforming strategies are also derived.

%In a general multicell environment, the aggregate channel vector between each user and multiple BSs, denoted as the \emph{network MIMO channel}, have \emph{non-identically} distributed elements due to a different pathloss to each BS. This makes the ergodic rate analysis of such systems tedious, since most of the available techniques assume channel vectors with independently and identically distributed (i.i.d.) elements. In this letter, we focus on the practical scenario of zero-forcing beamforming at the BSs under imperfect channel state information. We use the statistical properties of Gamma random variables to propose a new method for representing the network MIMO channel via an equivalent i.i.d. MIMO channel. We then employ the developed equivalency together with the available tools for i.i.d. MIMO channels to derive an accurate analytical. %Based on the developed i.i.d. equivalency, we provide new insights into the operation of network MIMO system together with a discussion of some non-trivial effects.
\end{abstract}

%\IEEEpeerreviewmaketitle
\vspace{-10pt}
%\vspace{-10pt}
\section{Introduction}\label{sec:intro}

\IEEEPARstart{I}{ncreasing} the area spectral efficiency of wireless networks requires a dense deployment of infrastructure and aggressive frequency reuse~\cite{ieeetc:6824752}. With shrinkage of the cell size, the number of cell edges -- and the number of cell-edge users -- in the network increases. This makes intelligent intercell interference (ICI) management crucial for successful operation of dense networks.

Multicell cooperation is an efficient technique to combat ICI~\cite{ieeetwc:5594708,ieeetc:Lozano2012}. In the most aggressive form of multicell cooperation, the channel state information (CSI) and the data of users are fully shared among base stations (BSs) via high-speed backhaul links. These BSs then act as a single distributed multi-antenna transmitter that serves multiple users through beamforming, commonly referred to as \emph{cooperative beamforming} or \emph{network multiple-input multiple-output} (MIMO). Although network MIMO can completely eliminate the ICI within the BSs' coverage area, it requires substantial signaling overhead and backhaul capacity for CSI and data sharing~\cite{ieeetwc:5757724}.

Prior work on network MIMO has mainly considered 2D cellular layouts and focused only on the horizontal plane of the wireless channel~\cite{ieeewc:5472938,ieeewc:6241389}, while ignoring the vertical dimension. Because of the 3D nature of the real-world wireless channel, employing network MIMO in the horizontal plane as the only ICI management strategy in the network seems like an inefficient and complex approach that cannot fully exploit all the degrees of freedom offered by the channel.

%the system level performance is not necessarily better because of the increased number of pilot signals required for CSI training, denoted as the \emph{training overhead}~\cite{ieeewc:Caire2010,ieeewc:6241389}. 

%In frequency-division duplexing (FDD) systems, such an overhead can degrade the performance of network MIMO even below that of conventional uncoordinated transmission when the number of coordinating antennas becomes very large~\cite{ieeewc:6095627}.

In a less complex form of multicell cooperation, commonly referred to as \emph{coordinated beamforming}, only the CSI of users is shared among the BSs to enable joint beamforming design, while the data for each user is transmitted by a single BS. With no need of data sharing, coordinated beamforming has significantly reduced signaling requirements compared to network MIMO. One simple coordinated beamfroming strategy is to exploit the vertical plane of the wireless channel for ICI management via coordinatively adapting the elevation angle of the BS antenna pattern, denoted as tilt. By appropriately selecting the tilt, it is possible to increase the desired signal level at an intended user, while reducing the ICI towards a non-intended user. In conventional cellular networks typically a fixed tilt is used at all BSs over all time-slots. This tilting strategy, denoted as \emph{cell-specific tilting}, can not adapt to the particular locations of the scheduled users. Therefore, users at different locations of the cell experience different antenna gains. For example, users close to the peak of the main beam observe a high antenna gain, while those close to the side-lobes experience a low antenna gain. With advances in antenna technology, it is possible to adapt the tilt rapidly using baseband processing~\cite{ieeewc:6181190}. This makes tilt an important parameter in the design of intercell interference management techniques. 

In this paper, we propose a hybrid multicell cooperation strategy that adaptively exploits either the horizontal plane or the vertical plane of the wireless channel to manage the ICI experienced by users, denoted as \emph{adaptive multicell 3D beamforming}. The key idea consists of partitioning the coverage area into disjoint ``vertical'' regions and adapting the multicell cooperation strategy to serve the users in each vertical region. CSI and data sharing are required only when the ICI management is performed in the horizontal plane. Our analytical and numerical results provide useful insights for the design of practical multicell cooperation strategies.%As a tool for our performance analysis, we use the statistical properties of Gamma random variable and derive accurate analytical expressions for user ergodic rates under imperfect CSI for different transmission modes.

\vspace{-12pt}
\subsection{Related Work}\label{subsec:relwork}

%Due to the training overhead bottleneck of network MIMO, recent work has focused on transmission techniques that switch between conventional uncoordinated transmission and network MIMO depending on the users' locations (see e.g.~\cite{ieeewc:5472938,ieeewc:6241389}). 

%This has not been extensively studied in the context of with BS coordination. In SCNs, the observed antenna gains at the users are comparable to their pathloss values owing to the reduced access distances to the BSs~\cite{ieeewc:6239994,ieeewc:Muller2012}. Therefore, antenna tilt has a more noticeable impact on users' performance in such networks.%Hence, cell-specific tilting seems to be an inefficient approach for optimizing the throughput over the coverage area in such networks.

%

%The value of the tilt in cell-specific tilting is chosen based on the maximization of some statistical performance metric such as \emph{edge throughput} or \emph{peak throughput} defined respectively as the $5$-percentile and the $95$-percentile of the throughput distribution over the cell.
%In practice there are limitations on the maximum tilt that can be applied and the speed with which the tilt can be varied~\cite{ieeewc:5493599,ieeewc:Halbauer2011}.
Dynamic tilt adaptation as an additional degree of freedom for designing efficient multicell cooperation techniques has
recently attracted a lot of interest~\cite{ieeewc:5285317,ieeewc:5910725,ieeewc:Halbauer2011,ieeewc:6218672}. The authors in~\cite{ieeewc:6218672}, have developed a coordinated beamforming framework in which intercell interference
is controlled solely in the vertical plane via joint adaptation of BSs' tilts to the locations of the users. This approach requires the knowledge of the users' locations at all the BSs, which is usually very difficult to obtain. An alternative tilt adaptation strategy is to divide the coverage area into so-called \emph{vertical regions} and use one out of a finite number of fixed tilts at the BS to serve each region. With this tilting strategy, which is denoted here as \emph{switched-beam tilting}, it is possible to increase the received signal power at a specific region in the desired cell, suppress the ICI at certain regions in the neighboring cells, or a combination thereof. In addition, switched-beam tilting does not require any knowledge about the locations of the users as a fixed tilt is applied to serve each vertical region. The work in~\cite{ieeewc:5285317,ieeewc:5910725,ieeewc:Halbauer2011} has studied different switched-beam tilting strategies for intercell interference avoidance. This work is, however, based on system-level simulations and does not provide any design guidelines, e.g., how to form vertical regions, how to choose the optimum tilts for different regions, or how to fairly schedule the transmission over different regions. In~\cite{ieeewc:Seifi2012,ieeetwc:6807764}, the authors investigated different methods to form vertical regions and determine the optimum tilt for each region in an isolated cell, but ICI was not considered. %Furthermore, only heuristic approaches were used to determine vertical regions' boundaries and the tilt value to be applied in each region.

One of the challenges in analyzing network MIMO performance in the presence of tilt is the lack of analytical performance measures. Available techniques for the ergodic rate analysis of MIMO systems mostly assume channel vectors with \emph{independently and identically distributed} (i.i.d.) elements, which simplifies the analysis significantly (see e.g.,~\cite{COML:5671564},~\cite{ieeecl:5953530},~\cite{ieeetwc:6880856}, and references therein). Such techniques, however, cannot be directly applied in the network MIMO setting. In this case each user might experience a different pathloss to each BS, and hence the elements of its aggregate channel vector to all BSs are \emph{non-identically} distributed in general. %This introduces an intrinsic asymmetry in the users' channel vector. Especially in practical systems, where the CSI is obtained via pilot-assisted estimation, this asymmetry appears in the estimated aggregate channel vector as well as the corresponding estimation error vector. This makes the ergodic rate analysis of such systems cumbersome.
In~\cite{ieeecl:5733435,ieeecl:J.Hoydis2012,ieeewc:Muller2012}, using the results from random matrix theory, a large system approximation for the ergodic sum-rate was derived for the uplink of a network MIMO system. Closed-form approximations for the user ergodic rate in downlink network MIMO transmission were derived in~\cite{ieeecl:5953530}. A limitation of~\cite{ieeecl:5733435,ieeecl:5953530,ieeecl:J.Hoydis2012,ieeewc:Muller2012} is that they assume perfect CSI at all BSs. Ergodic rate analysis of network MIMO under imperfect CSI has been recently investigated in~\cite{ieeewc:raey,ieeewc:6503474,ieeewc:6241389}. The work in~\cite{ieeewc:raey} considered the single-user scenario in which only one user is served using network MIMO. The author further assumed eigen-beamforming at each BS that makes the beamforming design independent of the pathloss of the user to the BSs. In this case, the ergodic rate analysis is performed readily using the existing techniques for i.i.d. MIMO channels. Analytical expression for user ergodic rate under multiuser transmission was derived in~\cite{ieeewc:6503474}. The author, however, did not consider the impact of pathloss, i.e., they assumed i.i.d. elements for the aggregate channel vector between each user and all BSs.
In~\cite{ieeewc:6241389}, a lower bound for the user ergodic rate was obtained for the special scenario in which users at fixed and symmetric set of locations in different cells are served. This symmetry causes the users to be \emph{statistically equivalent}, i.e., they experience the same set of pathloss values to all the BSs, which simplifies the ergodic rate analysis~\cite{ieeewc:6095627}. In practice, however, users are usually placed in asymmetric locations, which makes the ergodic rate analysis challenging.

%In~\cite{COML:5671564}, closed-form expressions for user ergodic rate under imperfect CSI were obtained in a single-cell setup without considering pathloss. Closed-form approximations for user ergodic rates with perfect CSI were derived in~\cite{ieeecl:5953530} for a multicell setup with a realistic pathloss model.
%The work in~\cite{YunRui2011} only considers the special case where a single user with equal pathloss to all BSs is served. In this special scenario the asymmetry in the network MIMO channel disappears and the problem transforms back to the i.i.d. case.
\vspace{-12pt}
\subsection{Contributions}\label{subsec:contb}

In this paper, our main aim is to design a practical multicell cooperation strategy that exploits both horizontal and vertical planes of the wireless channel to reduce the complexity and signaling requirements of network MIMO, while achieving comparable performance. Note that,
in our previous work~\cite{ieeetwc:6881264}, we studied the adaptive multicell 3D beamforming in a small-cell network under perfect CSI assumption. In contrast, in this paper, we focus on a more practical case of imperfect CSI, which includes the previous work~\cite{ieeetwc:6881264} as a special case. We also consider a dense macro cell network, instead of a small-cell network, that better matches our assumed propagation model. The main contributions of the paper are summarized as follows.

1) \textbf{Network MIMO ergodic rate analysis with imperfect CSI}: To make the analysis computationally efficient, we extend the results of~\cite{ieeecl:5953530} to the case of imperfect CSI. We propose a novel method to approximate the \emph{non-i.i.d.} network MIMO channel with an i.i.d. MIMO channel. We use this method together with the properties of Gamma random variables (RVs) to approximate the distributions of the desired signal and multiuser intracell interference power at each user. Using these distributions, we derive an accurate analytical expression for user ergodic rate under network MIMO transmission in the presence of imperfect CSI.

2) \textbf{Adaptive multicell 3D beamforming}: We first focus on cell-specific tilting and investigate the impact of tilt on the performance of conventional single-cell transmission and cooperative multicell transmission separately. Our analysis shows that cooperative multicell transmission is the preferred transmission strategy for users close to the cell boundary. For users close to the BSs, however, conventional single-cell transmission performs as good when the BSs \emph{coordinatively} apply large tilts. Based on these results, we propose a novel hybrid multicell cooperation strategy that divides the coverage area into disjoint vertical regions and adapts the transmission mode (i.e., single-cell vs. multicell) and the corresponding beamforming strategy when serving each vertical region. Our simulation results show that the proposed technique outperforms the conventional non-cooperative transmission while achieving comparable performance to that of network MIMO transmission with a significantly reduced signaling requirement.

\vspace{-12pt}
\subsection{Organization}\label{subsec:org}

The rest of the paper is organized as follows. Section~\ref{sec:sysmod} describes the system model including antenna patterns, propagation environment and received signal, transmission modes, and channel estimation. Beamforming techniques together with the ergodic rate expressions are presented in Section~\ref{subsec:transchem}. A new i.i.d. approximation for the network MIMO channel is introduced in Section~\ref{sec:EIID}. An analytical expression for user ergodic rate under network MIMO transmission is obtained in Section~\ref{sec:ergrate}. Section~\ref{sec:cellspctilt} studies cell-specific tilting for different transmission modes. The proposed adaptive multicell 3D beamforming is presented in Section~\ref{sec:hmtscheme}. Several approaches to apply the proposed technique to more general scenarios are discussed in Section~\ref{sec:multiplecoordclust}. Finally, Section~\ref{sec:concl} concludes the paper.

\vspace{-12pt}
\subsection{Notation}\label{subsec:not}

Scalars are denoted by lower-case letters. Vectors and matrices are denoted by bold-face lower-case and upper-case letters, respectively. $(\cdot)^{\hermit}$ is the complex conjugate transpose. $\mathbb{E}[\cdot]$ denotes the statistical expectation. $|\calS|$ is the cardinality of a set $\calS$. $\|\bx\|$ denotes the Euclidean norm of a complex vector $\bx$. 
%%%%%%%%%%%%%%%%%%%%%%%%%%%%%%%%%%%%%%%%%%%%%%%
%                 SYSTEM MODEL                %
%%%%%%%%%%%%%%%%%%%%%%%%%%%%%%%%%%%%%%%%%%%%%%%
\vspace{-12pt}
\section{System Model}\label{sec:sysmod}

We consider downlink transmission in a network consisting of a cluster of $B$ adjacent cells. Each cell has a multi-antenna BS located at a height $h_{\textsf{bs}}$ above the ground. We index all the $B$ cells in the network and their associated BSs by unique indices $b=1,\ldots,B$. There are $K$ users uniformly distributed over the coverage area and uniquely indexed as $k=1,\ldots,K$. Each user is at a height $h_{\textsf{u}}$ above the ground and has a single antenna. An example of a network consisting of three romb-shaped cells is shown in Fig.~\ref{fig:lincell}. We use this network configuration as an instructive example throughout the paper without loss of generality.

\begin{figure}[t]
\begin{minipage}[t]{1.0\linewidth}
\centering
%\psfrag{cellcellcellcell}[c][][0.8]{cell}
%\psfrag{hai2}[c][][0.7]{$h_{\textsf{bs}}$}
%\psfrag{hai1}[c][][0.7]{$h_{\textsf{u}}$}
%\includegraphics[width=0.65\columnwidth]{./images/sysmod.eps}
\def\svgwidth{0.65\columnwidth}
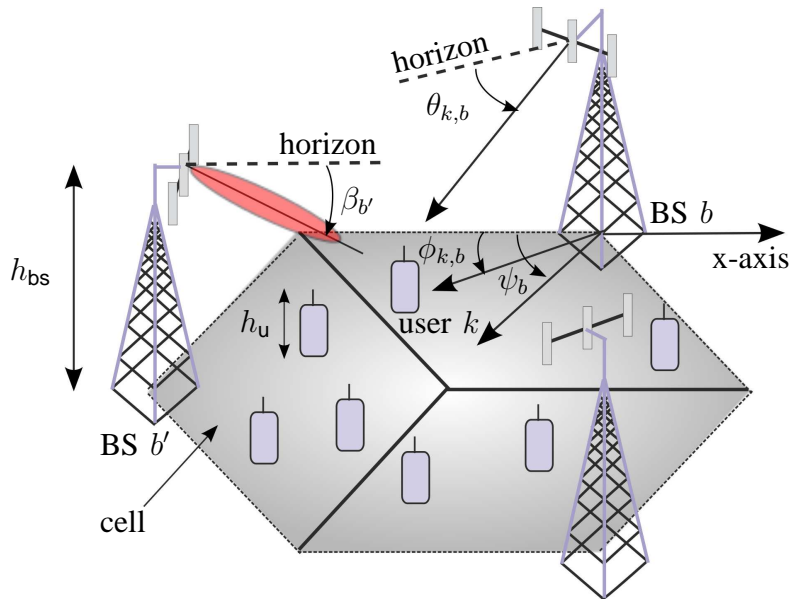
\caption{An example of a network consisting of $3$ adjacent romb-shaped cells with schematic illustration of spherical angles.}
\label{fig:lincell}
\end{minipage}
\end{figure}

\vspace{-12pt}
\subsection{Antenna Radiation Pattern}\label{sec:bsantpat}

For the antenna at the user, we assume a 3D unity-gain isotropic pattern. At the BS, we consider an array of $N_{\t}$ antennas that are arranged in a plane parallel to the ground. Each BS antenna itself comprises multiple vertically stacked radiating elements that are contained within a single radome. The pattern of each antenna depends on the number of radiating elements, their patterns, their relative positions, and their applied weights. By applying appropriate weights it is possible to control the vertical characteristics of the antenna pattern including the tilt. Here, to enable an abstraction of the role played by the radiating elements in controlling the tilt, we approximate each BS antenna pattern using the 3D directional model proposed in 3GPP~\cite[Section A.2.1.6.1]{ieeewc:3gpp}. We further assume that a common tilt is applied at all antennas of each BS. The observed antenna gain from any antenna of BS $b$ at user $k$ is expressed in dBi scale as
\begin{align}\label{eq:3dpat}
G_{k,b}^{\textsf{dBi}}(\beta_{b}) = 
-\min\left(\min\left[12\left(\frac{\phi_{k,b}-\psi_{b}}{\phi_{3\textsf{dB}}}\right)^{2},\mathrm{SLL}_{\az}\right] \right.
\left. +\min\left[12\left(\frac{\theta_{k,b} - \beta_{b}}{\theta_{3\textsf{dB}}}\right)^{2},\mathrm{SLL}_{\el}\right] \right),
\end{align}
where $\phi_{k,b}$ denotes the horizontal angle measured between the x-axis and the line in the horizontal plane connecting user $k$ to BS $b$, and $\theta_{k,b}$ is the vertical angle measured between the horizon and the line connecting user $k$ to BS $b$. In addition, $\psi_{b}$ represents the fixed orientation angle of BS $b$ array boresight relative to the x-axis, while $\beta_{b}$ denotes the tilt of BS $b$ measured between the horizon and the line passing through the peak of the main beam\footnote{Since the observed antenna gain at user $k$ is assumed to be the same from any antenna at BS $b$, for notational convenience we just use the BS index $b$ in~\eqref{eq:3dpat} instead of explicitly indexing each antenna of BS $b$.}. A schematic illustration of the spherical angles is shown in Fig.~\ref{fig:lincell}. 
Moreover, $\mathrm{SLL}_{\az}=25$ dB and $\mathrm{SLL}_{\el}=20$ dB are the side lobe levels (SLLs) in the horizontal and vertical planes of the BS antenna pattern, respectively. The half-power beamwidth (HPBW) in the horizontal and vertical planes are respectively denoted as $\phi_{3\textsf{dB}}=65^{\circ}$ and $\theta_{3\textsf{dB}}=6^{\circ}$. %Finally, $G_{\textsf{m}}=14$ dBi denotes the peak antenna gain. %Note that the model in~\eqref{eq:3dpat} is a simplified version of the commonly used \emph{Kathrein} 742215 antenna, which assumes a constant gain outside the main lobe instead of explicit side lobes~\cite{ieeewc:Gunnarsson2008}.

\vspace{-5pt}
\subsection{Propagation Environment and Received Signal Model}\label{subsec:propmod}
We focus on a typical non-line-of-sight (NLOS) propagation scenario in which the BS height is much greater than the large-scale clutter, such as buildings and trees, and the users are close to the ground (i.e., $h_{\mathsf{u}} \ll h_{\mathsf{bs}}$) and inside the clutter. In such a propagation environment, the horizontal and vertical planes of the wireless channel have different characteristics that should be taken into account when designing transmission strategies. In particular, the coverage area in the horizontal plane is relatively wider than that in the vertical plane. For example, in the romb-shaped cells in Fig. 1, the coverage area in the horizontal plane spans an angular range of $120^{\circ}$. But considering an cell radius of $150$ m, a BS height of $32$ m, and a user height of $1.5$ m, the vertical angles of the users located in $95\%$ of the cell area is less than $45^{\circ}$ and the vertical angle of a user at the cell edge is $11.5^{\circ}$. In this situation, the transmitted signal from the BS is more likely to experience a richer scattering in the horizontal plane than in the vertical plane. Hence, in this paper as a simple approach to model the propagation channel, we assume that multipath fading is rich in the horizontal plane, while it is negligible in the vertical plane. Although not fully realistic in a NLOS environment, recent detailed measurements have shown that this is a reasonable assumption when the BS is located high above the rooftop (as assumed in this paper) so that only few reflections and diffractions occur between the BS and the user in the vertical plane~\cite{ieeetc:Halbuer2013}.

Following the described propagation model, the channel between user $k$ and BS $b$ can be expressed as $\alpha_{k,b}(\beta_{b})\bh_{k,b}$. Here, $\alpha_{k,b}(\beta_{b})$ is the path gain given by~\cite{ieeewc:3gpp}
\vspace{-8pt}
\begin{equation}\label{eq:pl}
\alpha_{k,b}(\beta) = \mathrm{L}_{k,b} G_{k,b}(\beta_{b}),
\end{equation}
where $\mathrm{L}_{k,b}$ captures the distant-dependent pathloss between user $k$ and BS $b$, while $G_{k,b}(\beta_{b})$ indicates the observed antenna gain at user $k$ from BS $b$ and is given in dBi scale in~\eqref{eq:3dpat}. In addition, $\bh_{k,b} \in \bbC^{N_{\t} \times 1}$ denotes the small-scale fading channel vector between user $k$ and BS $b$. For simplicity of the analysis, we assume that the elements of $\bh_{k,b}$, $\forall k,b$, are i.i.d. $\calCN(0,1)$. 
We focus on universal frequency reuse and a narrowband frequency-flat fading channel. The complex base-band received signal at user $k$ can be expressed as
\begin{equation}\label{eq:recsig}
y_{k} = \sum_{b=1}^{B}\sqrt{\alpha_{k,b}(\beta_{b})} \, \bh_{k,b}^{\hermit}\bx_{b} + n_{k},
\end{equation}
where $\bx_{b} \in \bbC^{N_{\t} \times 1}$ is the transmitted signal from BS $b$, and $n_{k}$ indicates the normalized additive white Gaussian noise (AWGN) distributed as $\calCN(0,1)$.  

\vspace{-12pt}
\subsection{Downlink Transmission Modes}

In this work, we consider two well-known transmission modes that are adopted by the BSs when serving users:

\noindent \textbf{Conventional Single-cell Transmission (CST):} In this mode of transmission, each user is associated with one of the BSs to which it experiences the maximum average received power, denoted as the \emph{home} BS. The data to each user is transmitted by its home BS, while the transmissions from other BSs act as ICI. With no control information exchange among BSs, this transmission strategy leads to an \emph{uncoordinated} network, where adjacent cells operate independently and interfere mutually. Feasibility of control information exchange, however, facilitates \emph{coordinated} transmission at each BS to mitigate the ICI in the neighboring cells.

\noindent \textbf{Network MIMO Transmission (NMT):} In this transmission mode, the data of all users are shared among the BSs through high-speed backhaul links. The BSs then act as a single distributed multi-antenna transmitter with $BN_{\t}$ antennas to jointly serve the users in the coverage area. Under perfect CSI sharing among BSs, this transmission technique results in a fully \emph{cooperative} network in which the ICI can be completely removed.

\vspace{-12pt}
\subsection{Downlink Channel Estimation}

%We adopt a standard block-fading model~\cite{ieeewc:746779,ieeewc:6095627} in which channel vectors change independently from time-slot to time-slot. A time-slot corresponds to a block of length $T$ channel uses over which the small-scale fading coefficients are constant. We further focus on pilot-assisted downlink training and channel estimation in an FDD system. Each time-slot is divided into a training period of length $\tau$ channel uses and a data transmission period of length $T - \tau$ channel uses. At the beginning of each time-slot, each BS broadcasts some pilot symbols over a duration of $\tau$ channel uses to enable users to estimate their channel vectors to that BS.

We focus on pilot-based orthogonal channel training in all cells (corresponding to pilot reuse factor $B$) and minimum mean-square error (MMSE) channel estimation in a frequency-division duplexing (FDD) system. Under these assumptions, the canonical decomposition for the channel vector between user $k$ and BS $b$, i.e., $\bh_{k,b}$, can be expressed as~\cite{ieeewc:5773636,ieeewc:6095627}
\vspace{-8pt}
\begin{equation}\label{eq:channelest}
\bh_{k,b} = \hat{\bh}_{k,b} + \be_{k,b}.
\end{equation}
In~\eqref{eq:channelest}, $\hat{\bh}_{k,b}$ denotes the estimated channel vector of user $k$ from BS $b$ with the elements that are i.i.d. $\calCN(0,\kappa_{k,b}^{2}(\beta_{b}))$. In addition, $\be_{k,b}$ is the estimation error vector with the elements that are i.i.d. $\calCN(0,\sigma_{k,b}^{2}(\beta_{b}))$ where $\sigma_{k,b}^{2}(\beta_{b}) = 1/(1+\alpha_{k,b}(\beta_{b})BP)$ and $\kappa_{k,b}^{2}(\beta_{b}) = 1 - \sigma_{k,b}^{2}(\beta_{b})$. Each user estimates its channel vector(s) either to its home BS in CST or to all BSs in NMT and feeds back the estimated channel vector(s) to its home BS. We assume genie-aided feedback links which deliver the estimated channel vector(s) to the BSs perfectly. For the case of NMT, we further assume that these vectors are shared among BSs over error- and delay-free backhaul links to enable beamforming design.

\vspace{-10pt}
\section{Multiuser MIMO Beamforming and User Ergodic Rates}\label{subsec:transchem}

This section reviews the principles of linear multiuser MIMO zero-forcing beamforming in CST and NMT, and presents the expressions for the user ergodic rate at each transmission mode. %We resort to a popular choice for beamforming design where the estimated channel vectors at the BSs are treated as if they are perfect. Moreover, we consider equal power allocation among users to simplify the analysis\footnote{Although an important issue, spatial power allocation is beyond the scope of this paper and is left to our future work.}.} 

\vspace{-10pt}
\subsection{Beamforming and Ergodic Rates in CST}\label{subsec:bfcbt}
Let $\calK=\{1,\ldots,K\}$ denote the set of all users in the coverage area. In addition, $\calK_{b} \subseteq \calK$ is the set of users associated with BS $b$ under CST such that $\bigcup_{b=1}^{B} \calK_{b} = \calK$, $\calK_{b} \cap \calK_{b^{\prime}} = \emptyset$, $\forall b \neq b^{\prime}$, and $|\calK_{b}| \leq N_{\t}$. The last constraint will be easily satisfied in future densified networks as from one side more antennas will be deployed at each BS, i.e., larger $N_{\t}$, and from the other side fewer users will be served simultaneously per time-frequency resource block in each cell, i.e.,  smaller $|\calK_{b}|$. Now, let $\hat{\bH}_{b} \in \bbC^{N_{t} \times |\calK_{b}|}$ be the channel matrix having the estimated channel vectors of the users in cell $b$, i.e., $\{\hat{\bh}_{k,b}\}_{k \in \calK_{b}}$, as its columns. With only the knowledge of $\hat{\bH}_{b}$, BS $b$ designs the unit-norm zero-forcing beamformer for user $k \in \calK_{b}$, denoted as $\bw_{k,b} \in \calC^{N_{t} \times 1}$, such that $\hat{\bh}_{j,b}^{\hermit}\bw_{k,b} = 0$, $\forall j \neq k$ where $j \in \calK_{b}$. The transmitted signal $\bx_{b}$ in~\eqref{eq:recsigcbt} can now be expressed as
\vspace{-5pt}
\begin{equation}\label{eq:txsig}
\bx_{b} = \sum_{k \in \calK_{b}}\bw_{k,b}d_{k,b},
\end{equation}
where $d_{k,b}$ denotes the data symbol for user $k$. The transmitted signal from each BS is assumed to be subject to a power constraint $P$, i.e., $\bbE\left[\|\bx_{b}\|^{2}\right] = P$, $\forall b$. Using the canonical decomposition in~\eqref{eq:channelest}, we can now re-write the received signal for user $k \in \calK_{b}$ in~\eqref{eq:recsig} as
\vspace{-5pt}
\begin{align}\label{eq:recsigcbt}
y_{k} = & \underbrace{\sqrt{\alpha_{k,b}(\beta_{b})} \, \bh_{k,b}^{\hermit}\bw_{k,b}d_{k,b}}_\text{desired signal} + \underbrace{\sum_{\substack{ j \in \calK_{b} \\ j \neq k}}\sqrt{\alpha_{k,b}(\beta_{b})} \, \be_{k,b}^{\hermit}\bw_{j,b}d_{j,b}}_\text{intracell multiuser residual interference} \nonumber \\
& + \underbrace{\sum_{\substack{b^{\prime}=1 \\ b^{\prime} \neq b}}^{B}\sum_{\ell \in \calK_{b^{\prime}}} \sqrt{\alpha_{k,b^{\prime}}(\beta_{b^{\prime}})} \, \bh_{k,b^{\prime}}^{\hermit}\bw_{\ell,b^{\prime}}d_{\ell,b^{\prime}}}_\text{ICI} + n_{k}.
\end{align}
The signal-to-interference-plus-noise ratio (SINR) of user $k \in \calK_{b}$ is given by
\begin{align}\label{eq:sinrcbt}
&\gamma_{k,\textsf{CST}}(\bbeta) = \frac{\alpha_{k,b}(\beta_{b})\,\|\bh_{k,b}^{\hermit}\bw_{k,b}\|^{2}p_{k,b}}{1+\sum_{\substack{j \in \calK_{b} \\ j \neq k}}\alpha_{k,b}(\beta_{b}) \, \|\be_{k,b}^{\hermit}\bw_{j,b}\|^{2}p_{j,b} + \sum_{\substack{b^{\prime}=1 \\ b^{\prime} \neq b}}^{B}\sum_{\ell \in \calK_{b^{\prime}}}\alpha_{k,b^{\prime}}(\beta_{b^{\prime}}) \, \|\bh_{k,b^{\prime}}^{\hermit}\bw_{\ell,b^{\prime}}\|^{2}p_{\ell,b^{\prime}}},
\end{align}
where $p_{k,b}$ is the allocated power to user $k\in\calK_{b}$ by BS $b$ and $\bbeta = [\beta_{1}~\ldots~\beta_{B}]$ is the vector of applied tilts at all BSs.
For our tilt optimization in Section~\ref{sec:cellspctilt} (as it will be clarified later), we need to know the ergodic rate of user $k$ at any given location in cell $b$ given the path gain coefficients $\{\alpha_{k,b}(\beta_{b})\}_{b=1}^{B}$ at that location. Therefore, the desired performance metric we are interested in is the conditional ergodic rate given by
%\vspace{-5pt}
\begin{align}\label{eq:achratecbt}
R_{k,\textsf{CST}}(\bbeta) = \bbE\left[\log_{2}\left(1+\gamma_{k,\textsf{CST}}(\bbeta)\right)\bigg|\{\alpha_{k,b}(\beta_{b})\}_{b=1}^{B}\right].
\end{align}
Although optimal power allocation at each realization of small-scale fading can further improve the performance, it makes the analytical evaluation of~\eqref{eq:achratecbt} intractable. Thus, to facilitate a computationally efficient tilt optimization, we assume equal power allocation among users that enables a clean derivation of an accurate analytical expression for~\eqref{eq:achratecbt} (e.g., by using the techniques proposed in~\cite{ieeetwc:6880856,ieeecl:5953530}). We omit such a derivation here due to space limitations.

\vspace{-10pt}
\subsection{Beamforming and Ergodic Rates in NMT}\label{subsec:bfnmt} Define the aggregate channel vector $\bh_{k}$ from user $k$ to all BSs as
\begin{equation}\label{eq:comchan}
\bh_{k} = \left[\sqrt{\alpha_{k,1}(\beta_{1})} \,\bh_{k,1}^{\trans}\, ,\ldots,\, \sqrt{\alpha_{k,B}(\beta_{B})} \,\bh_{k,B}^{\trans}\right]^{\trans}.
\end{equation}
We refer to $\bh_{k}$ as the \emph{network MIMO channel vector} of user $k$ hereafter in the paper. Using the MMSE decomposition in~\eqref{eq:channelest}, the network MIMO channel vector in~\eqref{eq:comchan} can be written as
\vspace{-10pt}
\begin{equation}\label{eq:comchancandec}
\bh_{k} = \hat{\bh}_{k} + \be_{k},
\end{equation}
where $\hat{\bh}_{k}$ is the estimated network MIMO channel vector given by
\begin{equation}\label{eq:estcomchan}
\hat{\bh}_{k} = \left[\sqrt{\alpha_{k,1}(\beta_{1})} \,\hat{\bh}_{k,1}^{\trans}\, ,\ldots, \,\sqrt{\alpha_{k,B}(\beta_{B})} \,\hat{\bh}_{k,B}^{\trans}\right]^{\trans},
\end{equation}
and $\be_{k}$ denotes the network MIMO estimation error vector written as
\begin{equation}\label{eq:esterrcomchan}
\be_{k} = \left[\sqrt{\alpha_{k,1}(\beta_{1})} \,\be_{k,1}^{\trans}\, ,\ldots, \,\sqrt{\alpha_{k,B}(\beta_{B})} \,\be_{k,B}^{\trans}\right]^{\trans}.
\end{equation}
Now, let $\hat{\bH} \in \bbC^{BN_{t} \times |\calK|}$ be the channel matrix having the estimated network MIMO channel vectors of all users, i.e., $\{\hat{\bh}_{k}\}_{k \in \calK}$, as its columns. Assuming the knowledge of $\hat{\bH}$ at all BSs, the unit-norm zero-forcing beamformer $\bw_{k} \in \bbC^{BN_{t} \times 1}$ satisfies $\hat{\bh}_{j}^{\hermit}\bw_{k}=0$, for $\forall j \neq k$. 
\begin{remark}\label{rem:rem1} We notice that in CST the beamforming vectors in cell $b$ are solely determined from $\hat{\bH}_{b}$ which contains only channel vectors with i.i.d. elements. Such beamforming vectors can point in any direction in the complex space with equal probability and are commonly referred to as \emph{isotropically distributed} unit vectors~\cite{ieeewc:746779}. This phenomenon, however, does not hold in NMT because the estimated network MIMO channel vectors contain non-i.i.d. elements, resulting in \emph{non-isotropically distributed} beamforming vectors.
\end{remark}
The aggregate transmitted signal $\bx$ from all BSs can be expressed as
\vspace{-10pt}
\begin{equation}\label{eq:txsig}
\bx = \sum_{k \in \calK}\bw_{k}d_{k},
\end{equation}
where $d_{k}$ is the data symbol for user $k$. Here, we assume that $\bx$ is subject to a sum power constraint $BP$, i.e., $\bbE\left[\|\bx\|^{2}\right] = BP$, $\forall b$. While a per-BS power constraint is more relevant in practice, zero-forcing beamforming design with per-BS power constraint in NMT is computationally complex~\cite{ieeewc:5594707}. 
The received signal of user $k$ can be written as
\vspace{-10pt}
\begin{equation}\label{eq:recsignmt}
y_{k} = \underbrace{\bh_{k}^{\hermit}\bw_{k}d_{k}}_\text{desired signal} + \underbrace{\sum_{\substack{j \in \calK \\ j \neq k}}\be_{k}^{\hermit}\bw_{j}d_{j}}_\text{multiuser residual interference} + n_{k},
\end{equation}
and the SINR of user $k$ is expressed as
\begin{align}\label{eq:sinrnmt}
&\gamma_{k,\textsf{NMT}}(\bbeta) = \frac{\|\bh_{k}^{\hermit}\bw_{k}\|^{2}p_{k}}{1+\sum_{\substack{j \in \calK \\ j \neq k}}\|\be_{k}^{\hermit}\bw_{j}\|^{2}p_{j}}.
\end{align}
where $p_{k}$ is the allocated power to user $k$. Again, for tilt optimization in Section~\ref{sec:cellspctilt}, we need to compute the ergodic rate of user $k$ at any given location in the coverage area assuming the path gain coefficients $\{\alpha_{k,b}(\beta_{b})\}_{b=1}^{B}$ are known at that location. Our desired performance metric is the conditional ergodic rate defined as
\vspace{-5pt}
\begin{align}\label{eq:achratenmt}
R_{k,\textsf{NMT}}(\bbeta) = \bbE\left[\log_{2}\left(1+\gamma_{k,\textsf{NMT}}(\bbeta)\right)\bigg|\{\alpha_{k,b}(\beta_{b})\}_{b=1}^{B}\right].
\end{align}
Following the point in Remark~\ref{rem:rem1}, it holds that in NMT the beamformers at any time-slot depend on the particular realization of not only the small-scale fading, but also the path gains of all users. Therefore, for a given location of user $k$, the expectation in~\eqref{eq:achratenmt} should be taken with respect to all realizations of both the small-scale fading and the locations of other users. This makes the analytical evaluation of~\eqref{eq:achratenmt} very cumbersome if not impossible. To tackle this issue, we first assume an equal power allocation among users. We then provide a method in Section~\ref{sec:EIID} to approximate the network MIMO channel vector with an equivalent i.i.d. MIMO channel vector. Such an approximation will eliminate the asymmetry due to non-i.i.d. channels, thereby facilitating the use of existing results for i.i.d. channel vectors to compute $R_{k,\textsf{NMT}}(\bbeta)$. A review of some mathematical lemmas that prove useful in the following analysis is provided in Appendix~\ref{sec:mathprlm}.

\vspace{-12pt}
\section{Network MIMO Ergodic Rate Analysis}\label{sec:ergrate}
In this section, we derive an accurate analytical expression for the conditional ergodic rate in~\eqref{eq:achratenmt} assuming imperfect CSI.
\vspace{-12pt}
\subsection{A New I.I.D. Approximation for Network MIMO Channels}\label{sec:EIID}
We first propose a new method in which each user interprets its network MIMO channel vector as an i.i.d. channel vector with an equivalent path gain and an equivalent \emph{effective} degrees of freedom (DoF) per spatial dimension (to be defined later). Using~\eqref{eq:comchan}, we can write $\bh_{k}^{\hermit}\bh_{k} = \sum_{b=1}^{B}\alpha_{k,b}(\beta_{b})\bh_{k,b}^{\hermit}\bh_{k,b}$. It is well-known that $\bh_{k,b}^{\hermit}\bh_{k,b} = \|\bh_{k,b}\|^{2}$ is a chi-square RV with $2N_{\t}$ DoF scaled with $1/2$. Now, using Lemma~\ref{fact1} it holds that $\alpha_{k,b}(\beta_{b})\bh_{k,b}^{\hermit}\bh_{k,b} \sim \Gamma(N_{t},\alpha_{k,b}(\beta_{b}))$. Furthermore, using the independence of $\bh_{k,b}$ and $\bh_{k,b^{\prime}}$, $\forall b \neq b^{\prime}$, it follows that $\bh_{k}^{\hermit}\bh_{k}$ is a sum of independent Gamma RVs. According to Lemma~\ref{prop1}, any Gamma RV with the shape parameter $\mu_{k,\textsf{a}}(\bbeta)$ and the scale parameter $\theta_{k,\textsf{a}}(\bbeta)$, which are defined as
\begin{equation}\label{eq:eiidu}
\mu_{k,\textsf{a}}(\bbeta)  = N_{\t}\frac{(\sum_{b=1}^{B}\alpha_{k,b}(\beta_{b}))^{2}}{\sum_{b=1}^{B}\alpha_{k,b}^{2}(\beta_{b})}~~\text{and}~~\theta_{k,\textsf{a}}(\bbeta) = \frac{\sum_{b=1}^{B}\alpha_{k,b}^{2}(\beta_{b})}{\sum_{b=1}^{B}\alpha_{k,b}(\beta_{b})},
\end{equation}
has the same first and second moments as $\bh_{k}^{\hermit}\bh_{k}$. It can be easily shown that $N_{\t} \leq \mu_{k,\textsf{a}}(\bbeta) \leq BN_{\t}$, where the upper bound becomes exact when $\alpha_{k,1}(\beta_{1}) = \alpha_{k,2}(\beta_{2}) = \cdots = \alpha_{k,B}(\beta_{B})$, while the lower bound is attained when only one of the $\{\alpha_{k,b}(\beta_{b})\}_{b=1}^{B}$ is non-zero. Now, to define an equivalent i.i.d. channel for user $k$, we present a heuristic interpretation of $\mu_{k,\textsf{a}}(\bbeta)$ and $\theta_{k,\textsf{a}}(\bbeta)$.

Let $\bh_{k,\textsf{a}} \in \bbC^{BN_{\t} \times 1}$ be a channel vector for user $k$ with elements that are i.i.d. $\calCN(0,\theta_{k,\textsf{a}}(\bbeta))$. This can be looked at as if user $k$ is experiencing an i.i.d. channel with an equal path gain of $\theta_{k,\textsf{a}}(\bbeta)$ to all the coordinating BSs. Next, to incorporate the parameter $\mu_{k,\textsf{a}}(\bbeta)$ into the definition of $\bh_{k,\textsf{a}}$, we note that in the case where the user has an equal path gain to all the BSs, i.e., when $\theta_{k,\textsf{a}}(\bbeta)=\alpha_{k,1}(\beta_{1}) = \alpha_{k,2}(\beta_{2}) = \cdots = \alpha_{k,B}(\beta_{B})$, it holds that $\mu_{k,\textsf{a}}(\bbeta) = BN_{t}$ and hence $\bh_{k,\textsf{a}}^{\hermit}\bh_{k,\textsf{a}} \sim \Gamma(BN_{t},\theta_{k,\textsf{a}}(\bbeta))$. This means that from the perspective of this user each spatial dimension (or antenna) contributes one unit to the shape parameter of the resulting Gamma distribution of $\bh_{k,\textsf{a}}^{\hermit}\bh_{k,\textsf{a}}$. In the general case where the user experiences unequal path gains to the BSs, we have $N_{t} \leq \mu_{k,\textsf{a}}(\bbeta) < BN_{t}$. This can be looked at as if each spatial dimension in the channel offers $\frac{\mu_{k,\textsf{a}}(\bbeta)}{BN_{t}}$ ($\frac{1}{B} \leq \frac{\mu_{k,\textsf{a}}(\bbeta)}{BN_{t}} < 1$) unit to the shape parameter of the resulting Gamma distribution of $\bh_{k,\textsf{a}}^{\hermit}\bh_{k,\textsf{a}}$. We denote the parameter $\frac{\mu_{k,\textsf{a}}(\bbeta)}{BN_{t}}$ as the \emph{effective} DoF per spatial dimension at a given user location. Note that the effective DoF is just a notion introduced here to simplify our ergodic rate analysis and is completely different from the information theoretical DoF used in the MIMO system context (see e.g.~\cite{COML:5671564}). So from now on, we replace $\bh_{k}$ with $\bh_{k,\textsf{a}}$ in our analysis, i.e., we work with i.i.d. channel vectors, but whenever we need to consider the shape parameter for a Gamma RV in the ergodic rate analysis we consider the effective DoF per spatial dimension. This will be clarified in more details in the next section.

\vspace{-10pt}
\subsection{Conditional Ergodic Rate in NMT}\label{subsec:ergrateNMT}

To find an analytical expression for $R_{k,\textsf{NMT}}(\bbeta)$, we first replace the channel vector $\bh_{k}$ of user $k$ by its corresponding i.i.d. approximation $\bh_{k,\textsf{a}}$, $\forall k$, as defined in Section~\ref{sec:EIID}. With this replacement, we can perform the analysis in a transformed network that has a ``super'' cell with a layout equal to the whole coverage area of the original network (e.g., the hexagon in Fig.~\ref{fig:lincell}). The super cell consists of a ``super'' BS with $BN_{\t}$ antennas and $K$ users at the same locations as in the original network. The elements of the channel vector between user $k$ and the $BN_{\t}$ antennas of the super BS are i.i.d. $\calCN(0,\theta_{k,\textsf{a}}(\bbeta))$. In such a setup, the non-i.i.d. nature of the original network MIMO channel is captured via the notion of effective DoF per spatial dimension. 

Now, applying the well-known MMSE channel estimation model on $\bh_{k,\textsf{a}}$, we obtain a canonical decomposition as
%\vspace{-5pt}
\begin{equation}\label{eq:channelestagg}
\bh_{k,\textsf{a}} = \hat{\bh}_{k,\textsf{a}} + \be_{k,\textsf{a}}.
\end{equation}
In~\eqref{eq:channelestagg}, $\be_{k,\textsf{a}}$ is the estimation error vector with the elements that are i.i.d. $\calCN(0,\sigma_{k,\textsf{a}}^{2}(\bbeta))$, where $\sigma_{k,\textsf{a}}^{2}(\bbeta)=\theta_{k,\textsf{a}}(\bbeta)/(1+BP\theta_{k,\textsf{a}}(\bbeta))$ and $\hat{\bh}_{k,\textsf{a}}$ is the estimated channel vector with the elements that are i.i.d. $\calCN(0,\kappa_{k,\textsf{a}}^{2}(\bbeta))$, where $\kappa_{k,\textsf{a}}^{2}(\bbeta)=\theta_{k,\textsf{a}}(\bbeta)-\sigma_{k,\textsf{a}}^{2}(\bbeta)$.
Under these assumptions, the beamformer $\bw_{k,\textsf{a}}$, obtained by projection of vector $\hat{\bh}_{k,\textsf{a}}$ on the nullspace of the vectors $\{\hat{\bh}_{i,\textsf{a}}: \forall i \neq k\}$, is an isotropically distributed unit-norm vector. To evaluate~\eqref{eq:achratenmt}, we need the distributions of the desired signal term, i.e., $P_{k,\textsf{NMT}}^{\textsf{DS}}(\bbeta)=\|\bh_{k}^{\hermit}\bw_{k}\|^{2}BP/|\calK|$, and the multiuser residual interference term, i.e., $P_{k,\textsf{NMT}}^{\textsf{MRI}}(\bbeta)=\sum_{\substack{j \in \calK \\ j \neq k}}\|\be_{k}^{\hermit}\bw_{j}\|^{2}BP/|\calK|$. The authors in~\cite{ieeecl:5953530} have used a heuristic approach to approximate the distribution of $P_{k,\textsf{NMT}}^{\textsf{DS}}(\bbeta)$ with a Gamma distribution under perfect CSI. This approach, however, can not be directly extended to the case of imperfect CSI. In our proposed i.i.d. approximation, however, $\bh_{k,\textsf{a}}$, $\bw_{k,\textsf{a}}$, and $\be_{k,\textsf{a}}$ all have i.i.d. elements. This enables us to approximate the distributions of $\|\bh_{k}^{\hermit}\bw_{k}\|^{2}$ and $\|\be_{k}^{\hermit}\bw_{j}\|^{2}$ using respectively the distributions of $\|\bh_{k,\textsf{a}}^{\hermit}\bw_{k,\textsf{a}}\|^{2}$ and $\|\be_{k,\textsf{a}}^{\hermit}\bw_{j,\textsf{a}}\|^{2}$ together with the notion of effective DoF per spatial dimension as explained in the following.\\
\noindent\textbf{Desired signal term in NMT:} Using~\eqref{eq:channelestagg}, we extend the term $\|\bh_{k,\textsf{a}}^{\hermit}\bw_{k,\textsf{a}}\|^{2}$ as
%\vspace{-5pt}
\begin{align}\label{eq:dst}
\|\bh_{k,\textsf{a}}^{\hermit}\bw_{k,\textsf{a}}\|^{2} = \|(\hat{\bh}_{k,\textsf{a}}+\be_{k,\textsf{a}})^{\hermit}\bw_{k,\textsf{a}}\|^{2} \overset{(a)}{\approx}  \|\hat{\bh}_{k,\textsf{a}}^{\hermit}\bw_{k,\textsf{a}}\|^{2},
\end{align}
where $(a)$ follows by neglecting $\be_{k,\textsf{a}}^{\hermit}\bw_{k,\textsf{a}}$ as it is insignificant compared to $\hat{\bh}_{k,\textsf{a}}^{\hermit}\bw_{k,\textsf{a}}$. This is because for practical values of $\theta_{k,\textsf{a}}(\bbeta)BP$, $\forall k\in \calK$, we have $\sigma_{k,\textsf{a}}^{2}(\bbeta) \ll \kappa_{k,\textsf{a}}^{2}(\bbeta)$. Using Lemma~\ref{lem1:lem1}, it holds that $\hat{\bh}_{k,\textsf{a}}^{\hermit}\bw_{k,\textsf{a}}$ is equivalent to another vector of dimension $BN_{t}-|\calK|+1$ with the elements that are i.i.d. $\calCN(0,\kappa_{k,\textsf{a}}^{2}(\bbeta))$, and hence $\|\bh_{k,\textsf{a}}^{\hermit}\bw_{k,\textsf{a}}\|^{2} \sim \Gamma(BN_{t}-|\calK|+1,\kappa_{k,\textsf{a}}^{2}(\bbeta))$. We also note that the effective DoF per spatial dimension for user $k$ is equal to $\mu_{k,\textsf{a}}(\bbeta)/BN_{\t}$. Now, we propose to approximate the distribution of $\|\bh_{k}^{\hermit}\bw_{k}\|^{2}$ with $\Gamma((BN_{t}-|\calK|+1)\frac{\mu_{k,\textsf{a}}(\bbeta)}{BN_{\t}},\kappa_{k,\textsf{a}}^{2}(\bbeta))$, where the shape parameter is obtained by multiplying the shape parameter of the distribution of $\|\bh_{k,\textsf{a}}^{\hermit}\bw_{k,\textsf{a}}\|^{2}$, i.e., $BN_{t}-|\calK|+1$, with the effective DoF per spatial dimension, i.e., $\mu_{k,\textsf{a}}(\bbeta)/BN_{\t}$.
Next, using Lemmas~\ref{fact2}, it follows that $P_{k,\textsf{NMT}}^{\textsf{DS}}(\bbeta) \sim \Gamma(\mu_{k,\textsf{NMT}}^{\textsf{DS}}(\bbeta),\theta_{k,\textsf{NMT}}^{\textsf{DS}}(\bbeta))$, where
%\vspace{-5pt}
\begin{align}\label{eq:despar}
\mu_{k,\textsf{NMT}}^{\textsf{DS}}(\bbeta)= \left(BN_{\t}-|\calK|+1\right)\frac{\mu_{k,\textsf{a}}(\bbeta)}{BN_{\t}},~\theta_{k,\textsf{NMT}}^{\textsf{DS}}(\bbeta) = \frac{\kappa_{k,\textsf{a}}^{2}(\bbeta)BP}{|\calK|}.
\end{align}
\noindent\textbf{Multiuser residual interference term in NMT:} Using the independence of $\be_{k,\textsf{a}}$ and $\bw_{j,\textsf{a}}$, $\forall j\neq k$, and Lemma~\ref{lem1:lem1}, it holds that $\be_{k,\textsf{a}}^{\hermit}\bw_{j,\textsf{a}}$ is equivalent to another vector of dimension $1$ with an element distributed as $\calCN(0,\sigma_{k,\textsf{a}}^{2}(\bbeta))$, and hence $\|\be_{k,\textsf{a}}^{\hermit}\bw_{j,\textsf{a}}\|^{2} \sim \Gamma(1,\sigma_{k,\textsf{a}}^{2}(\bbeta))$. Noting that the effective DoF per spatial dimension is $\mu_{k,\textsf{a}}(\bbeta)/BN_{t}$, similar to the case of desired signal term, we approximate the distribution of $\|\be_{k}^{\hermit}\bw_{j}\|^{2}$ with $\Gamma(\frac{\mu_{k,\textsf{a}}(\bbeta)}{BN_{\t}},\sigma_{k,\textsf{a}}^{2}(\bbeta))$. Therefore, $P_{k,\textsf{NMT}}^{\textsf{MRI}}(\bbeta)$ is approximated as a sum of independent Gamma RVs with the same scale parameter. From Lemmas~\ref{fact2} and~\ref{fact3}, it results that $P_{k,\textsf{NMT}}^{\textsf{MRI}}(\bbeta)  \sim \Gamma(\mu_{k,\textsf{NMT}}^{\textsf{MRI}}(\bbeta),\theta_{k,\textsf{NMT}}^{\textsf{MRI}}(\bbeta))$, where
\begin{align}\label{eq:despar}
\mu_{k,\textsf{NMT}}^{\textsf{MRI}}(\bbeta) = \frac{(|\calK|-1)\mu_{k,\textsf{a}}(\bbeta)}{BN_{\t}},~~~\theta_{k,\textsf{NMT}}^{\textsf{MRI}}(\bbeta) = \frac{\sigma_{k,\textsf{a}}^{2}(\bbeta)BP}{|\calK|}.
\end{align}
The conditional ergodic rate in~\eqref{eq:achratenmt}, is now written as
\begin{align}\label{eq:achratenmtfin}
R_{k,\textsf{NMT}}(\bbeta) \approx &\,\bbE\left[\log_{2}(1+P_{k,\textsf{NMT}}^{\textsf{DS}}(\bbeta)+P_{k,\textsf{NMT}}^{\textsf{MRI}}(\bbeta))\bigg|\{\alpha_{k,b}(\bbeta)\}_{b=1}^{B}\right] \nonumber \\
&- \bbE\left[\log_{2}(1+P_{k,\textsf{NMT}}^{\textsf{MRI}}(\bbeta))\bigg|\{\alpha_{k,b}(\bbeta)\}_{b=1}^{B}\right].
\end{align}
To derive an analytical expression for~\eqref{eq:achratenmtfin}, we first use Lemma~\ref{prop1} to approximate $P_{k,\textsf{NMT}}^{\textsf{DS}}(\bbeta)+P_{k,\textsf{NMT}}^{\textsf{MRI}}(\bbeta)$ with another Gamma RV. After that, both expectation terms on the right hand side of~\eqref{eq:achratenmtfin} can be easily computed using Lemma~\ref{lemmeij}.

\vspace{-10pt}
\subsection{Numerical Example}\label{sec:numexp}
In this section, we verify the accuracy of the derived analytical expressions for CST and NMT via Monte-Carlo simulation. Our simulation parameters are set as follows. We use a 3D unit-gain isotropic pattern for each BS antenna. The cell radius, defined as the distance from the BS to one of the vertices of the romb-shaped cell, is set to $D=150$ m. BS and user heights are chosen as $h_{\textsf{bs}}=32$ m and $h_{\textsf{u}}=1.5$ m, respectively. For the pathloss factor $\mathrm{L}_{k,b}$ we use a standard distance-dependent model given by $\left(\frac{d_{k,b}}{D_{0}}\right)^{-\upsilon}$. Here, $d_{k,b}$ denotes the distance (in m) between user $k$ and BS $b$, accounting for the BS and the user heights, $D_{0}$ is a reference distance which is set to $1$ m, and $\upsilon$ is the pathloss exponent which is set to $3.76$. We further define the cell-edge SNR to be the SNR experienced at the edge of an isolated romb-shaped cell assuming maximum antenna gain. Throughout the paper, we set $N_{\t}=8$ and choose the BS transmit power $P$ so that the cell-edge SNR is $10$ dB. 

We move a sample user over the line segment connecting one of the BSs to the center of the hexagon in Fig.~\ref{fig:lincell}. For each location of the sample user, $|\calK|-1$ other users are uniformly distributed over the coverage area such that there are $|\calK_{1}|=|\calK_{2}|=|\calK_{3}|=6$ users in each cell. For CST, the ergodic rate of the sample user at a given location is obtained by averaging the instantaneous rate over $1000$ realizations of the small-scale fading for one random drop of the other $|\calK|-1$ users. For NMT, in addition to averaging over small-scale fading, we also perform another averaging over $100$ drops of the other $|\calK|-1$ users.

\begin{figure}[t]
\begin{minipage}[t]{1.0\linewidth}
\centering
\psfrag{xlabel}[c][b][0.8]{Distance from the BS (m)}
\psfrag{ylabel}[c][t][0.8]{Conditional ergodic rate [bps/Hz]}
%\psfrag{s4}[c][0.8]{\tiny STM=4}
\includegraphics[width=0.6\columnwidth]{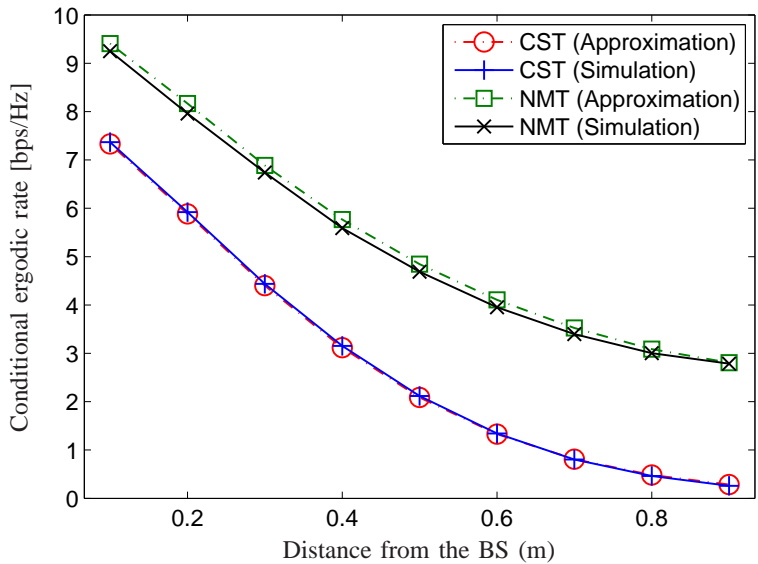}
\caption{Validation of the conditional ergodic rate approximations for a sample user moving on the line segment that connects a sample BS to the center of the hexagon in Fig.~\ref{fig:lincell}.}
\label{fig:ergrate}
\end{minipage}
\end{figure}

Figure~\ref{fig:ergrate} compares the conditional ergodic rate of the sample user obtained using the derived analytical expression and the Monte Carlo simulation for both CST and NMT. It can be easily seen that the match between theory and simulation is remarkably tight in CST. In NMT, a small mismatch is observed especially in areas close to the BS that can be attributed to the proposed i.i.d. approximations. We also mention that the match between theory and simulation in both CST and NMT is preserved when we change the number of users $|\calK|$. We, however, omit these results for brevity.

\vspace{-15pt}
\section{Cell-Specific Tilting in CST and NMT}\label{sec:cellspctilt}

In this section, we investigate the performance of the network in Fig.~\ref{fig:lincell} under cell-specific tilting. In this tilting strategy the applied tilts at the BSs are fixed at all times and does not adapt to the particular locations of users. Such tilts are usually found by maximizing some desired statistical performance metric which is independent of the particular realization of the users' locations. Adapting the tilt to the locations of the users, commonly referred to as \emph{user-specific tilting}, can potentially outperform cell-specific tilting~\cite{ieeetwc:6807764}. User-specific tilting, however, requires knowledge about users' locations (e.g., vertical angles), which is difficult to obtain. In this paper, we only focus on cell-specific tilting and leave the study of user-specific tilting to future work.

One popular approach for finding optimum tilts is the so-called \emph{throughput analysis}~\cite{ieeewc:5493599}. In this approach, three different performance metrics are used, namely the \emph{edge} throughput, \emph{average} throughput, and \emph{peak} throughput defined respectively as the 5-percentile, 50-percentile, and 95-percentile of the throughput cumulative distribution function (CDF) over the cell area. The throughput distribution for any given tilt is obtained by sampling the coverage area using a fine grid of user locations and computing the user throughput at each location by employing the analytical expressions derived in the previous section.

We let our simulation parameters follow those in Section~\ref{sec:numexp}. The main simulation assumptions are summarized in Table~\ref{table:simass}. Because of the symmetry in the considered network, we expect the optimal tilts for all BSs to be the same. Therefore, to reduce the search space we only consider the case where all BSs apply the same tilt, i.e.,  $\beta_{1}=\cdots=\beta_{B}=\beta$. 
We consider each transmission mode (i.e., CST\footnote{In this section, we focus on uncoordinated CST in which BSs operate independently.} or NMT) separately and save the throughput distributions for different values of $\beta$. We then use these throughput distributions to plot the edge, average, and peak throughput versus tilt as shown respectively in Figures~\ref{fig:edgethr},~\ref{fig:avethr}, and~\ref{fig:peakthr}.

\begin{table}[t]
\caption{Main simulation assumptions}
\centering
\begin{tabular}{l l l}
% centered columns (4 columns)
\hline
& Parameter & Modeling/value \\[0.5ex]
% inserts table
%heading
\hline % inserts single horizontal line
\multirow{5}{*}{\minitab[c]{Basic simulation \\ parameters}} & Cellular layout & Network model in Fig.~\ref{fig:lincell}\\
 & BS height, $h_{\textsf{bs}}$ & 32 m\\              % inserting body of the table
& User height, $h_{\textsf{u}}$ & 1.5 m\\
& Cell radius, $D$ & 150 m \\
& Pathloss, $L_{k,b}$ & $\left(\frac{d_{k,b}}{D_{0}}\right)^{-\upsilon}$, $D_{0}=1$, $\upsilon=3.76$\\
\hline
\multirow{4}{*}{\minitab[c]{3GPP antenna \\ parameters}} & Horizontal HPBW, $\phi_{3\textsf{dB}}$  & $65^{\circ}$ \\
& Vertical HPBW, $\theta_{3\textsf{dB}}$  & $6^{\circ}$ \\
& Horizontal SLL, $\mathrm{SLL}_{\az}$  & $25$ dB \\
& Vertical SLL, $\mathrm{SLL}_{\el}$  & $20$ dB \\[1ex]
\hline
%inserts single line
\end{tabular}
\label{table:simass}
% is used to refer this table in the text
\end{table}		

As can be seen in Fig.~\ref{fig:edgethr}, the maximum edge throughput is attained at $\beta=16^{\circ}$ for CST and $\beta=10^{\circ}$ for NMT. The edge throughput is mainly determined by users close to the cell edge. Therefore, to maximize the throughput of these users, the peak of the beam should be pointed more towards the edge of the cell. Moreover, we observe that the edge throughput at optimum tilt is significantly improved (by about $150\%$) in NMT compared to CST, showing that NMT is the best of the two transmission modes for edge throughput maximization.

The maximum average throughput is attained at $\beta=18^{\circ}$ for CST, while it is reached at $\beta=16^{\circ}$ for NMT as observed in Fig.~\ref{fig:avethr}. The larger tilt in CST is required to suppress the ICI, which is non-existent in NMT. When each transmission mode is operating at its optimum tilt, we observe that NMT improves the average throughput by $30\%$ compared to CST. 

In Fig.~\ref{fig:avethr}, we observe two peaks on the curves related to NMT. The first peak (on the left) occurs at a smaller tilt and corresponds to the scenario where there is an overlap among the main beams from all the BSs. In this case, users in the middle of each cell are not necessarily close to the peak of the main beam of any of the BSs. These users receive a significant part of their desired signal power from the neighboring BSs. The second peak (on the right) takes place at a larger tilt and relates to the case where the main beam from each BS is inside its own cell and is pointing directly to the users in the middle of the cell. Such users are close to the peak of the main beam of their home BS, and hence receive a major part of their desired signal from that BS. Interestingly, the second peak is larger than the first one, which shows the importance of tilt optimization in achieving the maximum performance gain in network MIMO. %Similar explanations can be given for the peaks observed in Fig.~\ref{fig:edgethr}, which is omitted here for brevity.

\begin{figure}[t]
\begin{minipage}[t]{0.5\textwidth}
\centering
\psfrag{xlabel}[c][b][0.8]{Tilt (degree)}
\psfrag{ylabel}[c][t][0.8]{Edge throughput [bps/Hz]}
%\psfrag{s4}[c][0.8]{\tiny STM=4}
\includegraphics[width=1.0\linewidth]{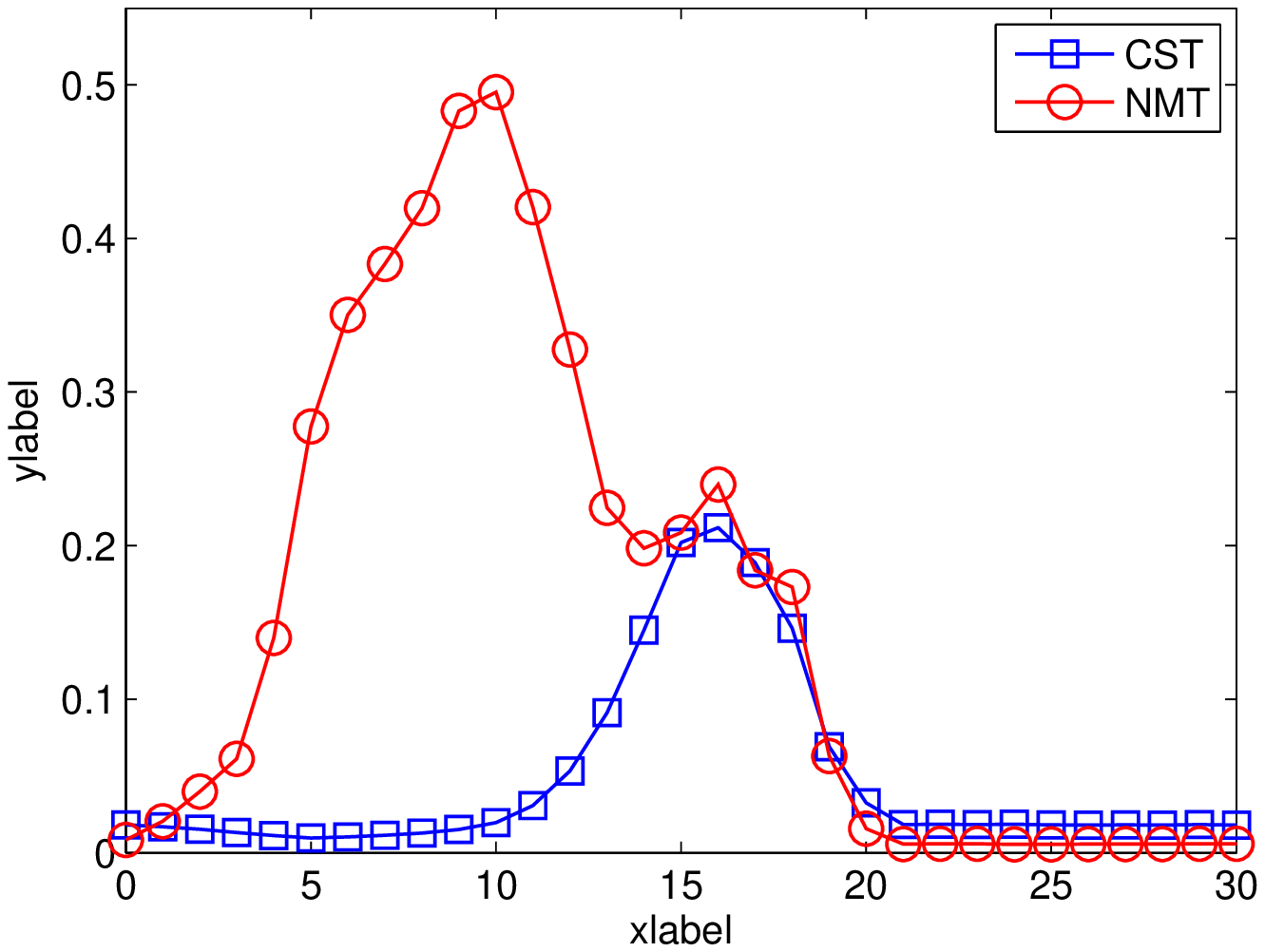}
\caption{Edge throughput comparison for CST and NMT.}
\label{fig:edgethr}
\end{minipage}
\begin{minipage}[t]{0.5\textwidth}
\centering
\psfrag{xlabel}[c][b][0.8]{Tilt (degree)}
\psfrag{ylabel}[c][t][0.8]{Average throughput [bps/Hz]}
%\psfrag{s4}[c][0.8]{\tiny STM=4}
\includegraphics[width=1.0\linewidth]{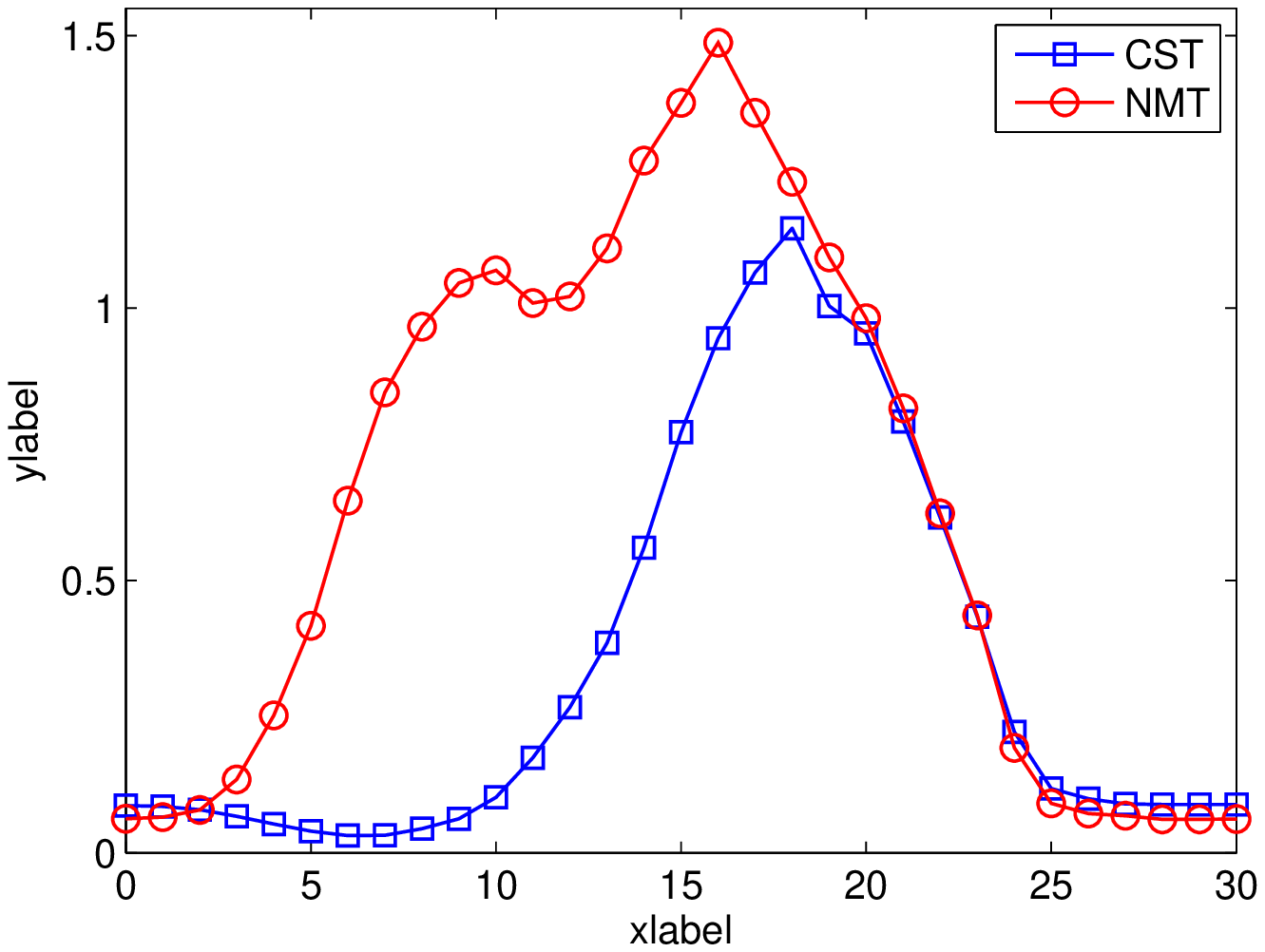}
\caption{Average throughput comparison for CST and NMT.}
\label{fig:avethr}
\end{minipage}
%\hskip 5pt
\end{figure}
\begin{figure}[h]
\centering
\begin{minipage}[t]{0.5\linewidth}
\centering
\psfrag{xlabel}[c][b][0.8]{Tilt (degree)}
\psfrag{ylabel}[c][t][0.8]{Peak throughput [bps/Hz]}
%\psfrag{s4}[c][0.8]{\tiny STM=4}
\includegraphics[width=1.0\columnwidth]{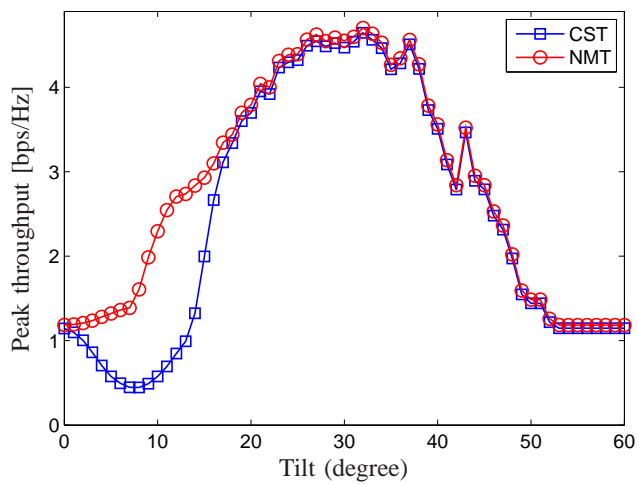}
\caption{Peak throughput comparison for CST and NMT.}
\label{fig:peakthr}
\end{minipage}
\end{figure}

The peak throughput is usually attained by users close to the BS. So we expect the maximizing tilt for peak throughput to be larger than that for average throughput. This is verified in Fig.~\ref{fig:peakthr}, where we see that the optimum tilt for both CST and NMT is $32^{\circ}$. Interestingly, the peak throughput performance of CST is almost the same as that of NMT for a tilt greater than $18^{\circ}$. We note that the allocated power to each user is the same in both CST and NMT. In addition, large tilt already provides enough protection against ICI for users close to the BS. Hence, it seems that serving such a user with a given power from the closest BS with the strongest channel in the presence of sufficiently suppressed ICI performs equally good as serving it with the same power from multiple geographically BSs with disparate channel strengths.

Notice that Figures~\ref{fig:edgethr},~\ref{fig:avethr}, and~\ref{fig:peakthr} have been plotted assuming $6$ users per cell. This choice is motivated by the results in~\cite{COML:5671564} indicating that the number of users per cell for multiuser MIMO transmission should be less than $N_{\t}$ in the presence of imperfect CSI. Our results from separate simulations, however, show that the number of users has a negligible effect on the optimum tilts for different performance metrics. These results are omitted here due to space limitations.

One important highlight here is that with cell-specific tilting, it is not possible to optimize all the performance metrics, i.e., edge, average, and peak throughput, at the same time as each of these metrics is maximized at a different tilt. One low-complexity solution to address this issue is switched-beam tilting in which at each time-slot one out of a set of finite tilts is applied at the BSs to increase the received signal power at a specific region in the desired cell, suppress the ICI at certain regions in the neighboring cells, or a combination thereof. In the next section, we exploit the idea of switched-beam tilting and propose a novel transmission strategy that is capable of achieving a \emph{tradeoff} in maximizing all performance metrics simultaneously.

\vspace{-12pt}
\section{Adaptive Multicell 3D Beamforming}\label{sec:hmtscheme}

In the previous section, NMT was shown to be the best transmission mode for edge throughput maximization. Furthermore, the peak throughput performance of CST was shown to be almost as good as NMT for tilts greater than some threshold ($18^{\circ}$ for the considered scenario). For average throughput, we observed that NMT has moderate superiority over CST.

The aforementioned argument eventually leads to the following hypothesis: a multicell cooperation strategy that would serve the users in different regions of the cell, namely, the cell-interior region or the cell-edge region, with an appropriate transmission mode, i.e., CST or NMT, and a corresponding appropriate tilt could potentially achieve a tradeoff in maximizing all the performance metrics simultaneously. Hence, we propose a hybrid multicell cooperation strategy, denoted as \emph{adaptive multicell 3D beamforming}, with the following components:
\begin{enumerate}
\item A division of the coverage area in Fig.~\ref{fig:lincell} into two disjoint \emph{vertical regions} as follows: i) a cell-interior region that consists of three disjoint vertical regions each associated with one of the BSs. The vertical region in cell $b$ is obtained as the intersection of the coverage area with a circle of radius $D_{b,\textsf{int}}$ centered at BS $b$; ii) a cell-edge region that is shared among all BSs and is obtained by removing the cell-interior region from the coverage area. This is illustrated in Fig.~\ref{fig:hybmod}.

\item A transmission technique that at each time-slot serves either the cell-interior region using CST or the cell-edge region using NMT. We emphasize that only one of these transmission modes can be active at each time-slot. A switched-beam tilting strategy is also employed that applies at each BS $b$, $\forall b$, a fixed tilt $\beta_{b,\textsf{CBT}}$ when serving the cell-interior region (see Fig.~\ref{fig:hybmodcbt}), or a fixed tilt $\beta_{b,\textsf{NMT}}$ when serving the cell-edge region (see Fig.~\ref{fig:hybmodnmt}). In line with the definition of cell-specific tilting, we denote this switched-beam tilting strategy as \emph{region-specific tilting}, as the tilts applied at the BSs to serve each region are independent of the particular locations of the users in that region.

\item  A scheduler to share the available time-slots between the cell-interior region and the cell-edge region. We define the cell-interior region activity factor $\nu_{\textsf{CST}}$ ($0\leq\nu_{\textsf{CST}}\leq1$) as the fraction of the total time-slots in which the cell-interior region is served. Similarly, $\nu_{\textsf{NMT}}$ ($0\leq\nu_{\textsf{NMT}}\leq1$) is the cell-edge region activity factor such that $\nu_{\textsf{CST}}+\nu_{\textsf{NMT}}=1$. The parameters $\nu_{\textsf{CST}}$ and $\nu_{\textsf{NMT}}$ depend, for any realization of users' locations, on the number of users in each region and hence are denoted as \emph{user-specific} parameters.
\end{enumerate}
\begin{remark}\label{rem:rem1} We highlight that the proposed adaptive multicell 3D beamforming contains different elements of coordination among the BSs. For example, the BSs need to \emph{coordinatively} serve the same vertical region at each time slot, i.e., either the cell-interior region or the cell-edge region. In addition, in CST, BSs exploit the vertical plane of the wireless channel to perform ICI suppression as well as throughput optimization via \emph{coordinatively} applying sufficiently large tilts. Each BS, however, uses the horizontal plane independently for multiuser MIMO transmission within the vertical region in its own cell. In NMT, BSs use the vertical plane for throughput optimization via \emph{coordinatively} applying suitable tilts, while they exploit the horizontal plane for ICI mitigation through \emph{cooperative beamforming}.
\end{remark}

In the proposed technique, the parameters $\{D_{b,\textsf{int}},\beta_{b,\textsf{CST}},\beta_{b,\textsf{NMT}}\}_{b=1}^{B}$, denoted as \emph{region-specific parameters}, and the user-specific parameters $\nu_{\textsf{CBT}}$ and $\nu_{\textsf{NMT}}$ are unknown and need to be determined. Next, we present the methods to determine these parameters.

\begin{figure}
\centering
\subfigure[CST mode]{
\def\svgwidth{0.46\columnwidth}
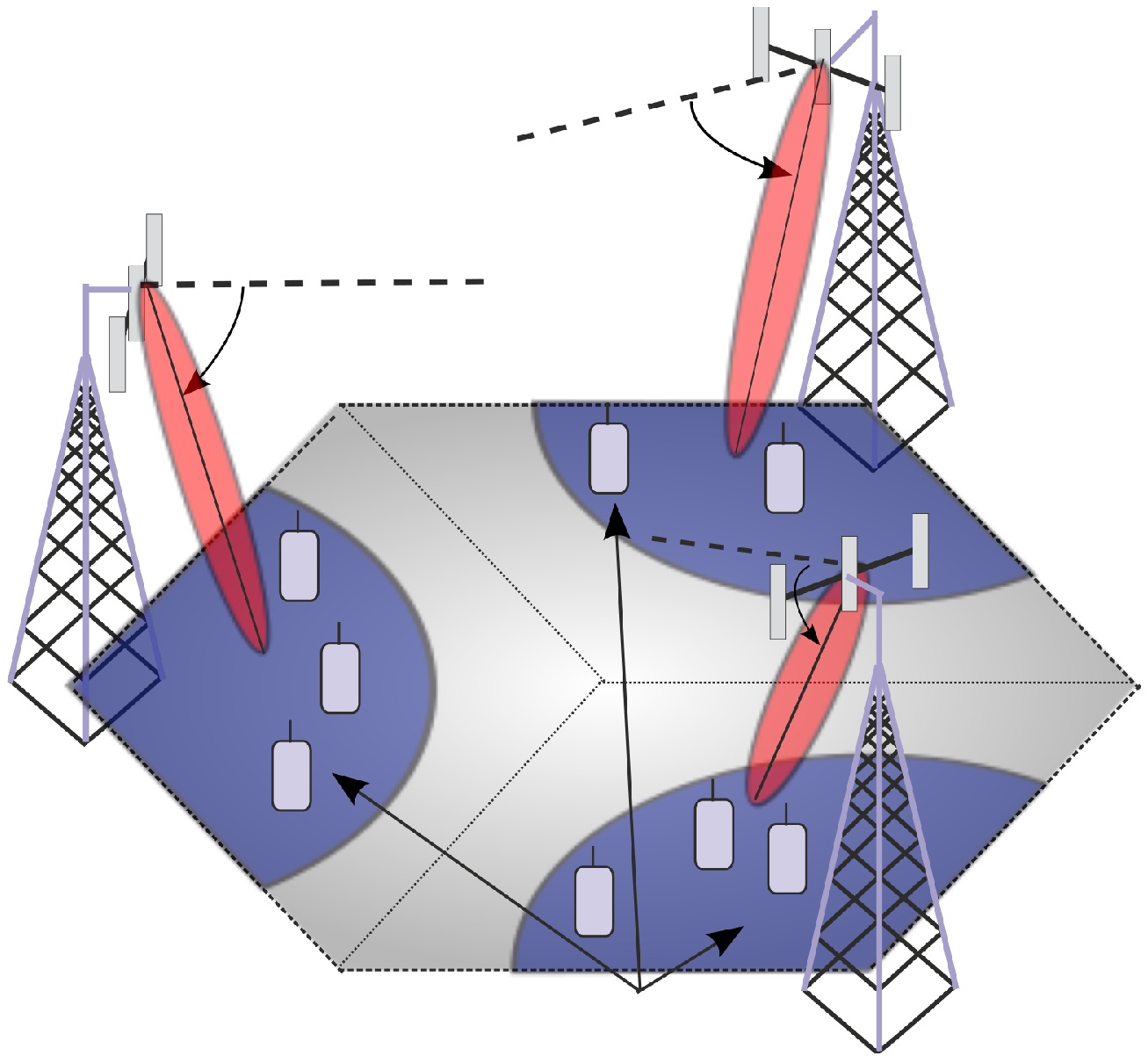
\label{fig:hybmodcbt}
}
\subfigure[NMT mode]{
\def\svgwidth{0.46\columnwidth}
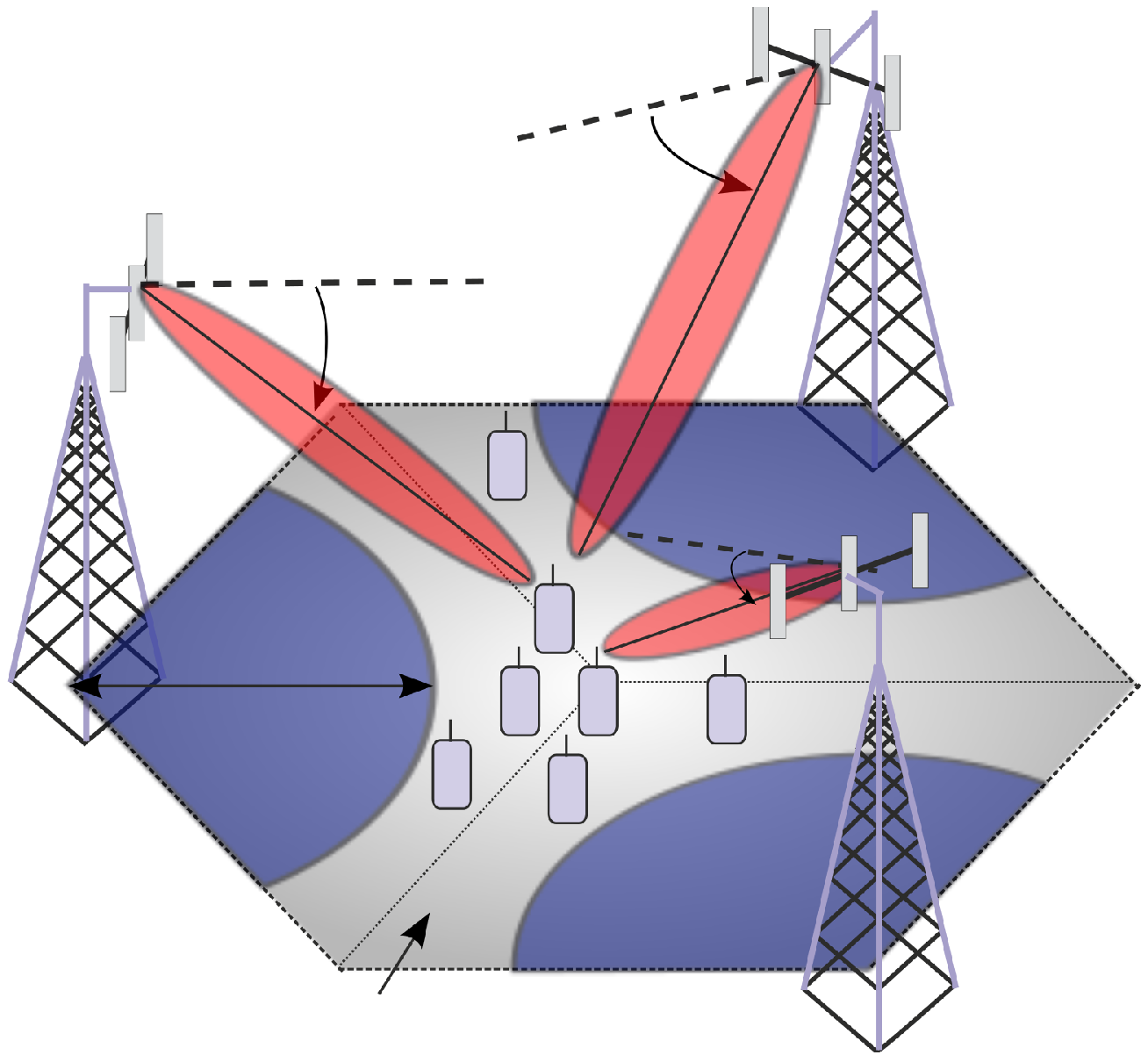
\label{fig:hybmodnmt}
}
\caption{Schematic illustration of the transmission modes and the corresponding beam tilting strategies in the proposed adaptive multicell 3D beamforming.}
\label{fig:hybmod}
\end{figure}

%\vspace{-8pt}
\subsection{Determining Region-Specific Parameters}\label{subsec:detpar}

To determine $\{D_{b,\textsf{int}},\beta_{b,\textsf{CST}},\beta_{b,\textsf{NMT}}\}_{b=1}^{B}$, we focus on average throughput maximization. Note that in the proposed adaptive multicell 3D beamforming users in the cell-interior region are served by CST and users in the cell-edge region are served by NMT. Therefore, we expect the edge and peak throughput performance to be at a satisfactory level for appropriate choices of $\{D_{b,\textsf{int}},\beta_{b,\textsf{CST}},\beta_{b,\textsf{NMT}}\}_{b=1}^{B}$ that would maximize the average throughput. 

Without loss of generality, we focus on the symmetric network in Fig.~\ref{fig:lincell}, where we expect the region-specific parameters at all cells to be the same. We further drop the cell index and denote these parameters hereafter as $D_{\textsf{int}}$, $\beta_{\textsf{CST}}$, and $\beta_{\textsf{NMT}}$. The average throughput for any given $D_{\textsf{int}}$, $\beta_{\textsf{CST}}$, and $\beta_{\textsf{NMT}}$, denoted as $\bar{R}(D_{\textsf{int}},\beta_{\textsf{CST}},\beta_{\textsf{NMT}})$, is determined using the users' throughput both in the cell-interior region and in the cell-edge region. Although the user throughput in both regions can be obtained using the analytical expressions for~\eqref{eq:achratecbt} and~\eqref{eq:achratenmt}, it is very difficult to draw any insight about how the average throughput changes with these parameters. In the following we provide a heuristic discussion about this issue.
 
On one hand, if $D_{\textsf{int}}$ becomes too small, most of the users in the interior part of the cell are served using NMT. Since $\beta_{\textsf{NMT}}$ is set so that the peak of the beam is pointing more towards the cell-edge, many of these users are close to the side-lobe of the antenna beam. Such users could potentially achieve a higher throughput if they were served by CST with $\beta_{\textsf{CST}} > \beta_{\textsf{NMT}}$. In that case, they would be both closer to the peak of the beam of their home BS and very well protected against ICI (see Fig.~\ref{fig:hybmodcbt}). On the other hand, if $D_{\textsf{int}}$ becomes too large, most of the users in the vicinity of the cell edge are served by CST using $\beta_{\textsf{CST}}$. Such users will experience a low throughput as they are both close to the side-lobe of the antenna pattern of the their home BS and subject to a large ICI. As a result, the optimum $D_{\textsf{int}}$ that maximizes the average throughput is expected to be somewhere in the middle of the cell. 

To determine the optimum $D_{\textsf{int}}$, $\beta_{\textsf{CST}}$ and $\beta_{\textsf{NMT}}$, we simulate $\bar{R}(D_{\textsf{int}},\beta_{\textsf{CST}},\beta_{\textsf{NMT}})$ for $D_{\textsf{int}}\in[0.15D,0.95D]$. For any given $D_{\textsf{int}}$, we exhaustively search over $\left[\arctan(\frac{h_{\mathsf{bs}}-h_{\mathsf{u}}}{D_{\textsf{int}}}),90^{\circ}\right]$ to find the optimum value of $\beta_{\textsf{CST}}$ and over $\left[0^{\circ},\arctan(\frac{h_{\mathsf{bs}}-h_{\mathsf{u}}}{D_{\textsf{int}}})\right]$ to find the optimum value of $\beta_{\textsf{NMT}}$ that maximizes $\bar{R}(D_{\textsf{int}},\beta_{\textsf{CST}},\beta_{\textsf{NMT}})$. Our simulation setup is the same as in Section~\ref{sec:cellspctilt}. In Fig.~\ref{fig:coordis} the average throughput is plotted versus the normalized cell-interior region radius, i.e., $D_{\textsf{int}}/D$. As can be seen in the figure, the maximum is achieved at $D_{\textsf{int}}=0.6D$ and by using a corresponding $\beta_{\textsf{CST}}=21^{\circ}$ and $,\beta_{\textsf{NMT}}=14^{\circ}$.

\begin{figure}[t]
\begin{minipage}[t]{1.0\linewidth}
\centering
\psfrag{b1}[c][b][0.8]{$\beta_{\textsf{CST}}=21^{\circ}$, $\beta_{\textsf{NMT}}=14^{\circ}$}
\psfrag{xlabel}[c][b][0.8]{Normalized cell-interior region radius, $D_{\textsf{int}}/D$}
\psfrag{ylabel}[c][t][0.8]{Average throughput, $\bar{R}(D_{\textsf{int}},\beta_{\textsf{CST}},\beta_{\textsf{NMT}})$, [bps/Hz]}
%\psfrag{s4}[c][0.8]{\tiny STM=4}
\includegraphics[width=0.7\columnwidth]{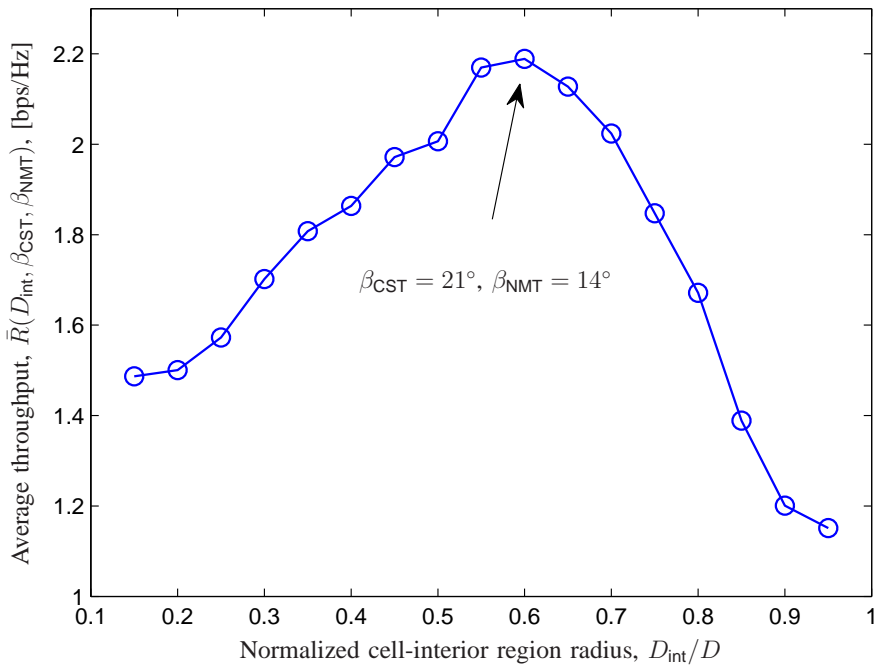}
\caption{Optimization of average throughput with respect to $D_{\textsf{int}}$, $\beta_{\textsf{CST}}$, and $\beta_{\textsf{NMT}}$.}
\label{fig:coordis}
\end{minipage}
\end{figure}

\vspace{-10pt}
\subsection{Determining User-Specific Parameters}\label{sec:fs}

%So far we have only considered the users that are served at each time-slot, denoted as the \emph{active users}. The number of such users in the network is limited by the available resources that can be used for transmission. For example, under zero-forcing beamforming, the BSs in NMT can serve at most $BN_{\textsf{t}}$ users. In practice the number of available users in the cell-edge region might be larger than $BN_{\textsf{t}}$. Hence, a scheduler is employed to select a subset of the available users at each time-slot. 
Let $\calK_{\textsf{CST}}$ and $\calK_{\textsf{NMT}}$ denote the set of all users in the cell-interior region and cell-edge region, respectively. We focus on one drop of $|\calK_{\textsf{CST}}|+|\calK_{\textsf{NMT}}|$ users over the coverage area and assume that the path gain coefficients of all users are known. Under user scheduling, the throughput of user $k$ is defined as
\begin{equation}\label{eq:actfac}
R_{k}^{\textsf{sch}} = \left \{
 \begin{array}{ll}
  \nu_{\textsf{CST}}\,R_{k,\textsf{CST}}^{\textsf{sch}}(\beta_{\textsf{CST}}) & \textrm{if } k\in\calK_{\textsf{CST}}\\
  \nu_{\textsf{NMT}}\,R_{k,\textsf{NMT}}^{\textsf{sch}}(\beta_{\textsf{NMT}}) & \textrm{if } k\in\calK_{\textsf{NMT}}.\\
  \end{array} \right .
\end{equation}
%where $\nu_{\textsf{CBT}}$ ($0\leq\nu_{\textsf{CBT}}\leq1$) denotes the cell-interior region activity factor, i.e., the fraction of the total time-slots in which the cell-interior region is active. Similarly, $\nu_{\textsf{NMT}}$ ($0\leq\nu_{\textsf{NMT}}\leq1$) is the cell-edge region activity factor such that $\nu_{\textsf{CBT}}+\nu_{\textsf{NMT}}=1$. 
Here, $R_{k,\textsf{CST}}^{\textsf{sch}}(\beta_{\textsf{CST}})$ and $R_{k,\textsf{NMT}}^{\textsf{sch}}(\beta_{\textsf{NMT}})$ indicate the \emph{user per region} throughput for user $k$ in the cell-interior region and cell-edge region, respectively. $R_{k,\textsf{CST}}^{\textsf{sch}}(\beta_{\textsf{CST}})$ ($R_{k,\textsf{NMT}}^{\textsf{sch}}(\beta_{\textsf{NMT}})$) is obtained by averaging the instantaneous rate over all the time-slots in which the cell-interior region (cell-edge region) is active. Note that user $k\in\calK_{\textsf{CST}}$ ($k\in\calK_{\textsf{NMT}}$) might not necessarily be served at each time-slot in which the cell-interior region (cell-center region) is active, in which case its instantaneous rate is zero. Therefore, $R_{k,\textsf{CST}}^{\textsf{sch}}(\beta_{\textsf{CST}})$ and $R_{k,\textsf{NMT}}^{\textsf{sch}}(\beta_{\textsf{NMT}})$ are in general different from the conditional ergodic rates defined in Section~\ref{sec:ergrate}. To determine $\nu_{\textsf{CST}}$ and $\nu_{\textsf{NMT}}$, the scheduler has to solve the following convex optimization problem:
%\vspace{-20pt}
\begin{align}\label{eq:cvxopt}
\text{maximize}&~~g(\bR^{\textsf{sch}}) \nonumber \\
\text{subject to}&~~R_{k}^{\textsf{sch}} \leq \left \{
 \begin{array}{ll}
  \nu_{\textsf{CST}}\,R_{k,\textsf{CST}}^{\textsf{sch}}(\beta_{\textsf{CST}}) & \textrm{if } k\in\calK_{\textsf{CST}}\\
  \nu_{\textsf{NMT}}\,R_{k,\textsf{NMT}}^{\textsf{sch}}(\beta_{\textsf{NMT}}) & \textrm{if } k\in\calK_{\textsf{NMT}}\\
  \end{array} \right ., \nonumber \\
&~~\nu_{\textsf{CST}}+\nu_{\textsf{NMT}}= 1, \nu_{\textsf{CST}}, \nu_{\textsf{NMT}} \geq 0.
\end{align}
In~\eqref{eq:cvxopt}, $g(\cdot)$ is a concave and componentwise non-decreasing utility function with a suitable
notion of fairness~\cite{ieeewc:Mo00} and $\bR^{\textsf{sch}}$ denotes the vector of throughputs of all users in the coverage area.
Here, we focus on the popular choice of proportional fair scheduling~\cite{ieeewc:Mo00} whose utility function is given by
\vspace{-5pt}
\begin{equation}\label{eq:utilpfs}
g(\bR^{\textsf{sch}}) = \sum_{k\in\calK_{\textsf{CST}}}\log(R_{k}^{\textsf{sch}})+\sum_{k\in\calK_{\textsf{NMT}}}\log(R_{k}^{\textsf{sch}}).
\end{equation}
Solving~\eqref{eq:cvxopt} using the utility function in~\eqref{eq:utilpfs}, we obtain the activity factors of the cell-interior region and the cell-edge region as
\vspace{-10pt}
\begin{equation}\label{eq:actfacregion}
\nu_{\textsf{CST}}=\frac{|\calK_{\textsf{CST}}|}{|\calK_{\textsf{CST}}|+|\calK_{\textsf{NMT}}|},~~\nu_{\textsf{NMT}}=\frac{|\calK_{\textsf{NMT}}|}{|\calK_{\textsf{CST}}|+|\calK_{\textsf{NMT}}|}.
\end{equation}
Notice that for proportional fair scheduling the values of $\nu_{\textsf{CST}}$ and $\nu_{\textsf{NMT}}$ are independent of $\{R_{k,\textsf{CST}}^{\textsf{sch}}(\beta_{\textsf{CST}})\}_{k\in\calK_{\textsf{CST}}}$ and $\{R_{k,\textsf{NMT}}^{\textsf{sch}}(\beta_{\textsf{NMT}})\}_{k\in\calK_{\textsf{NMT}}}$, which are usually difficult to compute analytically. Calculating the activity factors of vertical regions for other utility functions is beyond the scope of this paper and is left to our future work.
\begin{remark}\label{rem:rem1} We emphasize that in adaptive multicell 3D beamforming, the region-specific parameters are obtained via offline analysis and remain unchanged once they are determined. User-specific parameters, however, depend on the number (and not the location) of users in each vertical region for any particular realization of users' locations and need to be updated when the number of users per region changes. Therefore, the operation of the proposed scheme does not require any knowledge about the users' locations in the network.
\end{remark}

\subsection{Numerical Results}\label{sec:numres}

In this section, the performance of the proposed adaptive multicell 3D beamforming is evaluated via Monte Carlo simulation. Our simulation parameters follows those in Section~\ref{sec:cellspctilt}. We use a drop-based simulation, where at each drop $8$ users are randomly placed in each romb-shaped cell. The users are associated with the cell-interior region or the cell-edge region based on their locations in the cell. The time-slots are shared among the vertical regions according to~\eqref{eq:actfacregion}. The users in each vertical region are served assuming standard proportional fair user selection, multiuser MIMO zero-forcing in Section~\ref{subsec:transchem}, and spatial waterfilling power allocation~\cite{ieeewc:Caire2010}. We simulate a sufficient number of small-scale fading realizations such that all users achieve their limiting throughputs. We then stack the users' throughputs over all drops to obtain the throughput distribution over the coverage area. 

We compare the performance of five different systems as follows. 1) CST with edge throughput maximizing tilt ($\beta=16^{\circ}$), denoted as Uncoord-CST-E; 2) CST with average throughput maximizing tilt ($\beta=18^{\circ}$), denoted as Uncoord-CST-A; 3) NMT with edge throughput maximizing tilt ($\beta=10^{\circ}$), denoted as NMT-E; 4) NMT with average throughput maximizing tilt ($\beta=16^{\circ}$), denoted as NMT-A; and 5) Adaptive multicell 3D beamforming with $\beta_{\textsf{CST}}=21^{\circ}$ and $\beta_{\textsf{NMT}}=14^{\circ}$, denoted as AM-3D-BF. By design the chosen tilts\footnote{Note that peak throughput maximizing tilt is not considered for comparison as it is not relevant in practice.} are obtained using the throughput analysis in Sections~\ref{sec:cellspctilt} and~\ref{subsec:detpar} that assume equal power allocation among users and, are independent of any particular user scheduling algorithm. 

Figure~\ref{fig:cdfcomp} compares the throughput CDF of different transmission strategies. As can be seen, AM-3D-BF significantly improves the throughput over the coverage area compared to both Uncoord-CST-E and Uncoord-CST-A. It also provides moderate throughput gain over NMT-E and NMT-A in major parts of the coverage area except the cell boundary region, where it actually underperforms NMT-E. It is also observed that for CST (NMT), it is not possible to maximize the edge, average, and peak throughput simultaneously as Uncoord-CST-A (NMT-A) achieves a higher average and peak throughput than Uncoord-CST-E (NMT-E), but experiences a lower edge throughput. The proposed AM-3D-BF, however, achieves a tradeoff in simultaneously maximizing all the three performance metrics.

\begin{figure}[t]
\begin{minipage}[t]{1.0\linewidth}
\centering
\psfrag{xlabel}[c][b][0.8]{Throughput [bps/Hz]}
\psfrag{ylabel}[c][t][0.8]{Throughput CDF [\%]}
\centering
\includegraphics[width=0.7\columnwidth]{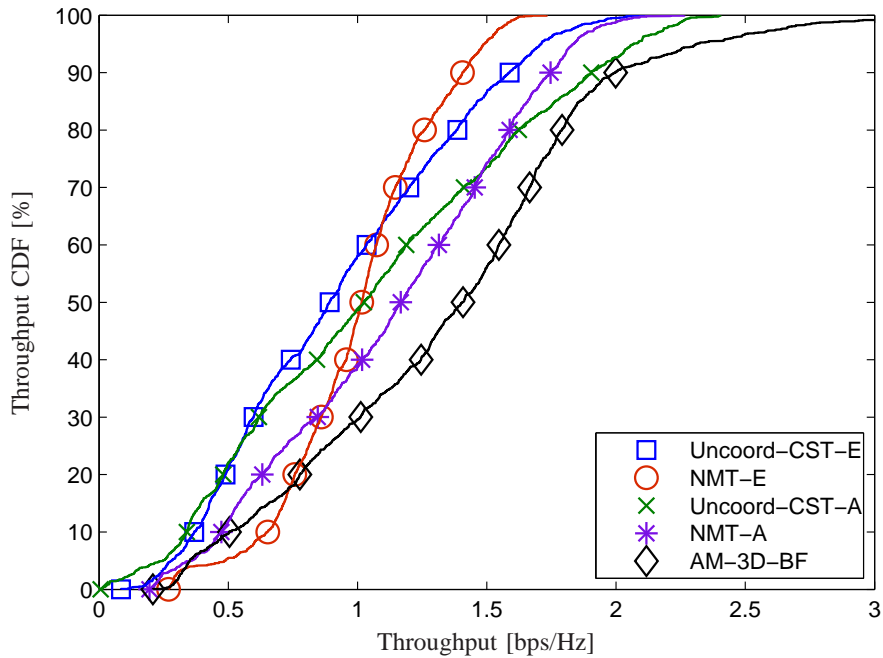}
\caption{Comparison of throughput CDF over the coverage area for different transmission strategies.}
\label{fig:cdfcomp}
\end{minipage}
\end{figure}

Figure~\ref{fig:perfgaincomp} shows the edge, average, and peak throughput gain of AM-3D-BF over the other four comparative systems. It is seen that compared to other systems, AM-3D-BF experiences at most $28\%$ loss in edge throughput (compared to NMT-E), while it provides at least $20\%$ average throughput gain (compared to NMT-A) and $12\%$ peak throughput gain (compared to Uncoord-CST-A). The loss in edge throughput is because of the larger tilt used to serve the cell-edge region in AM-3D-BF compared to NMT-E ($14^{\circ}$ versus $10^{\circ}$). The gain in peak throughput results from using a relatively large tilt ($21^{\circ}$) when serving the users close to the BSs. The gain in average throughput, however, seems to come from the joint transmission mode and tilt adaptation that improves the throughputs of users in the middle of the cell as these users are located partly in the cell-interior region and partly in the cell-edge region.

\begin{figure}[t]
\begin{minipage}[t]{1.0\linewidth}
\centering
\psfrag{xlabel}[c][b][0.8]{Percentile of throughput CDF}
\psfrag{ylabel}[c][t][0.8]{Throughput gain [\%]}
%\psfrag{s4}[c][0.8]{\tiny STM=4}
\includegraphics[width=0.7\columnwidth]{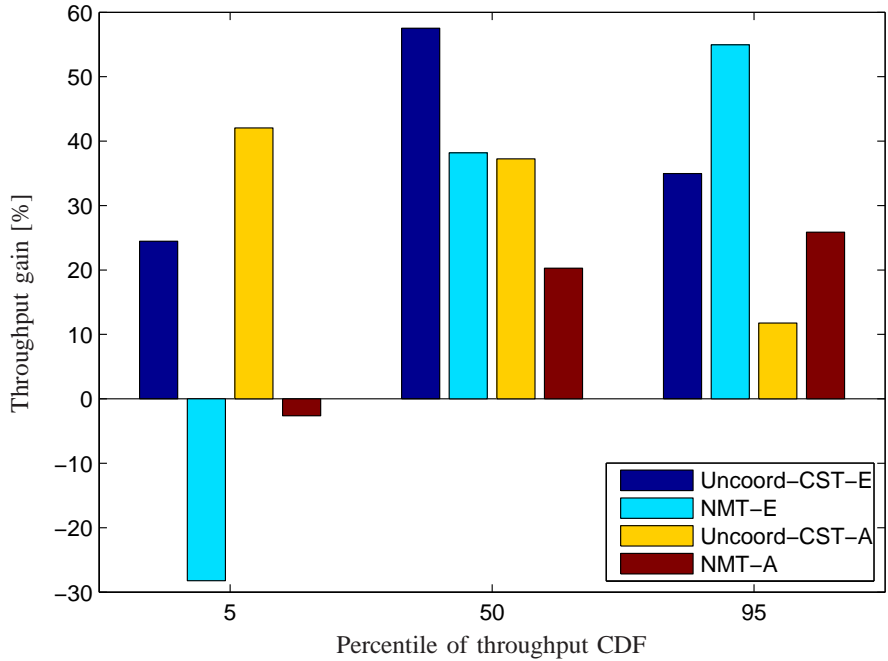}
\caption{Performance gain comparison of the proposed adaptive multicell 3D beamforming.}
\label{fig:perfgaincomp}
\end{minipage}
\end{figure}

It is also worth mentioning that assuming proportional fair scheduling and waterfilling power allocation does not seem to change the relative behavior of CST and NMT compared to that in Section~\ref{sec:cellspctilt} with equal power allocation and no user scheduling. For example, in Fig.~\ref{fig:cdfcomp}, we observe that NMT-E achieves a lower peak throughput than Uncoord-CST-E. Similarly, in Fig.~\ref{fig:peakthr}, the peak throughput of NMT at $\beta=10^{\circ}$ is smaller than that of CST at $\beta=16^{\circ}$. Therefore, the employed throughput analysis seems like a reasonable approach to determine the optimum tilts.  

Finally, we highlight that in the proposed adaptive multicell 3D beamforming a fraction $A_{\textsf{CST}}/A_{\textsf{cov}}$ of the coverage area is served using CST, where $A_{\textsf{CST}}$ is the area of the cell-interior region and $A_{\textsf{cov}}$ denotes the area of the whole network. In the considered system setup, $A_{\textsf{CST}}=\pi(0.6D)^{2}$ and $A_{\textsf{cov}}=3\sqrt{3}D^{2}/2$, resulting in $\approx 44\%$ of the coverage area to be served using CST. Equivalently, this means that the proposed technique requires about $44\%$ less signaling and data sharing overhead compared to NMT, while achieving a comparable performance.

\section{Extension to Large Cellular Networks}\label{sec:multiplecoordclust}
So far, we have designed and evaluated the adaptive multicell 3D beamforming for an isolated cluster of $B=3$ mutually interfering BSs shown in Fig.~\ref{fig:lincell}. In this section, we suggest possible approaches to employ the adaptive multicell 3D beamforming in networks with a large number of cells. A detailed investigation of the proposed approaches is beyond the scope of this paper, and should be studied in a separate work.

In cellular networks with directive antenna patterns at the BSs, a major part of the ICI is usually generated by a small number of neighboring BSs. For example, in a network with hexagonal cellular layout and $120^{\circ}$ cell sectoring, the cluster configuration in Fig.~\ref{fig:lincell} already includes the strongest interfering BSs. Hence, multicell cooperation within this cluster can mitigate a significant part of the overall ICI~\cite{ieeewc:Muller2012}. In fact, adding more cells to this cluster might bring diminishing gain or even loss when taking into account the excessive overhead and complexity associated with CSI acquisition at the BSs (training, estimation, and feedback)~\cite{ieeewc:6095627}. Thus, one possible approach to employ the proposed adaptive multicell 3D beamforming in large cellular networks is via static clustering in which the network is divided into fixed and disjoint clusters such that each cluster contains the dominant interfering BSs~\cite{DBLP:journals/ejwcn/LiBS12,ieeewc:6095627}. Figure~\ref{fig:multiclust} shows an exemplary $21$-cell network that has been divided into $7$ disjoint clusters, denoted as $\text{C1},\text{C2},\ldots,\text{C7}$. In this network, the proposed adaptive multicell 3D beamforming can effectively control the intra-cluster interference in a large part of the cluster area with reasonable overhead and complexity. The performance of users at the cluster edge might, however, be still limited by out-of-cluster interference. In the following, without loss of generality, we focus on $\text{C1}$ in the exemplary network in Fig.~\ref{fig:multiclust}. We further classify the out-of-cluster interfering BSs of $\text{C1}$ into two groups and propose possible approaches to suppress the interference from each group.

\begin{figure}[t]
\begin{minipage}[t]{1.0\linewidth}
\centering
%\psfrag{cellcellcellcell}[c][][0.8]{cell}
%\psfrag{hai2}[c][][0.7]{$h_{\textsf{bs}}$}
%\psfrag{hai1}[c][][0.7]{$h_{\textsf{u}}$}
%\includegraphics[width=0.65\columnwidth]{./images/sysmod.eps}
\def\svgwidth{0.5\columnwidth}
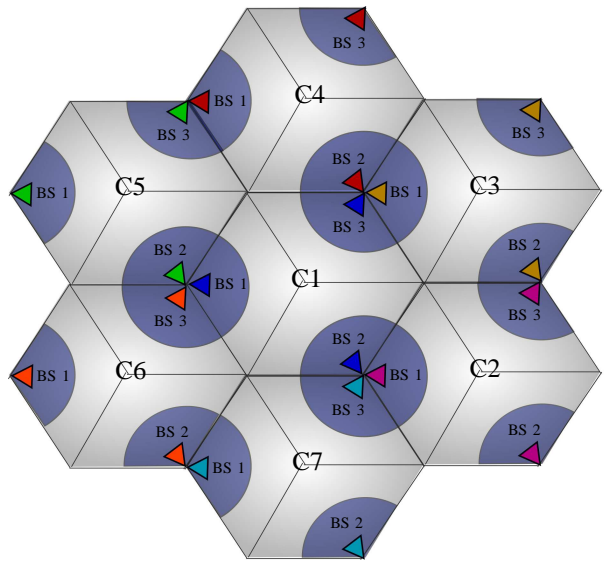
\caption{Illustration of employing the proposed adaptive multiple 3D beamforming in a $21$-cell network.}
\label{fig:multiclust}
\end{minipage}
\end{figure}

The first group of out-of-cluster interfering BSs consists of those BSs in the neighboring clusters that are installed next to the BSs in $\text{C1}$ (i.e., BS $1$ in $\text{C2}$ and in $\text{C3}$, BS $2$ in $\text{C4}$ and in $\text{C5}$, and BS $3$ in $\text{C6}$ and in $\text{C7}$). The interference from this group is generated mainly through their back lobes, i.e., the side lobes in the back side of the antenna. Moreover, due to the significantly smaller antenna gain of the back lobes compared to that of the main lobe, the interference from this group will be most detrimental to their adjacent cell in $\text{C1}$ and not to all cells. For example, the interference from BS $1$ in $\text{C3}$ and BS $2$ in $\text{C4}$ will mainly impact the users in cell $3$ of $\text{C1}$. One way to suppress such interference would be to multiplex the time-slots allocated to serve the cell-interior region and cell-edge region across neighboring clusters such that for each BS in $\text{C1}$ and its two adjacent BSs in the neighboring clusters, at least one will serve a different vertical region (i.e., cell-interior region or cell-edge region) than the others. Moreover, because CSI sharing among adjacent BSs does not incur extra backhaul usage, another way to suppress the out-of-cluster interference from this group would be via some form of coordinated beamforming among each BS in $\text{C1}$ and its two adjacent BSs in the neighboring clusters.

The second group consists of all other BSs in $\text{C2}$ to $\text{C7}$. These BSs are located at a distance larger than $D$ from the users in $\text{C1}$ and generate out-of-cluster interference mostly through their main lobes. The interference from this group might be significant only when they serve the cell-edge regions in their corresponding cluster using $\beta_{\textsf{NMT}}$. Thanks to the directivity of the antenna patterns and the small vertical HPBW, in networks with a dense deployment of BSs, such interference can be significantly suppressed by e.g., slightly increasing $\beta_{\textsf{NMT}}$. Note that when these group of BSs serve the cell-interior region of their corresponding clusters, sufficient out-of-cluster interference mitigation is achieved via using $\beta_{\textsf{CST}}>\beta_{\textsf{NMT}}$.

We finally highlight that there are also other clustering techniques, such as semi-static clustering~\cite{ieeetwc:5959253,ieeetwc:6884162}, dynamic clustering\cite{ieeetwc:4533793,ieeetwc:6827165}, and hierarchical clustering~\cite{ieeewc:1657802} among others, that provide more flexibility to deal with out-of-cluster interference. The tradeoffs between different techniques are complicated by the need for more CSI training and feedback overhead, and additional backhaul signaling. Detailed investigation of these approaches is left to future work.

\vspace{-10pt}
\section{Concluding Remarks}\label{sec:concl}

In this paper, we investigated downlink transmission in a cellular network with small number of cells that employs two well-known transmission modes, namely, conventional single-cell transmission and fully cooperative multicell transmission denoted as network MIMO. To facilitate a computationally efficient analysis, we proposed a novel method for approximating the non-i.i.d. network MIMO channel with an equivalent i.i.d. MIMO channel. We used this method to derive an accurate analytical expression for the user ergodic rate under network MIMO transmission with imperfect CSI. We then considered directional antennas with vertically adjustable beams at the BSs  and studied cell-specific tilting for the two transmission modes separately. Our results showed that upon applying sufficiently large tilts at the BSs, the two transmission modes have similar performance in regions close to the BSs. In fact, with sufficient interference isolation provided by applying large tilts, receiving a given desired signal power from the closest BS with the strongest channel seems to perform as good as receiving it from multiple geographically distributed BSs with disparate channel strengths. Using this conclusion, we proposed an adaptive multicell 3D beamforming technique that adaptively exploits the horizontal and vertical planes of the wireless channel for interference management as well as throughput optimization. The proposed technique divides the coverage area into two vertical regions and adapt the multicell cooperation strategy, including the transmission mode and beamforming strategy at the BSs, when serving each region. Numerical results showed the superiority of the proposed technique over the uncoordinated conventional single-cell transmission. The proposed technique also seems to provide a superior performance-complexity tradeoff compared to network MIMO transmission. Finally, we presented possible approaches to employ the proposed adaptive multicell 3D beamforming in networks with a large number of cells.

%%%%%%%%%%%%%%%%%%%%%%%%%%%%%%%%%%%%%%%%%%%%%%%%%%%%%%%%%%
%                      Appendices
%%%%%%%%%%%%%%%%%%%%%%%%%%%%%%%%%%%%%%%%%%%%%%%%%%%%%%%%%%
\vspace{-10pt}
\appendices
\section{Mathematical Lemmas}\label{sec:mathprlm}

In this appendix, we provide some well-known lemmas that prove useful in the analyses in Sections~\ref{sec:EIID} and~\ref{sec:ergrate}.

\begin{lemma}\label{fact2}
If $Y$ is a Gamma RV with shape parameter $\mu$ and scale parameter $\theta$, i.e., $Y \sim \Gamma(\mu,\theta)$, and $b$ is a positive constant, then $bY \sim \Gamma(\mu,b\theta)$.
\end{lemma}
\begin{lemma}\label{fact3}
If $Y_{i} \sim \Gamma(\mu_{i},\theta)$ for $i=1,\ldots,N$, then $\sum_{i=1}^{N}Y_{i} \sim \Gamma\left(\sum_{i=1}^{N}\mu_{i},\theta\right)$.
\end{lemma}
\begin{lemma}\label{fact1}
If $Z$ is a chi-square RV with $2r$ DoF, denoted as $Z \sim {\chi}_{2r}^{2}$, and $a$ is a positive constant, then $aZ \sim \Gamma(r,2a)$.
\end{lemma}
\begin{lemma}[Muirhead~\cite{Muirhead1982}]\label{lem1:lem1}
The projection of an $M$-dimensional vector with i.i.d. $\calCN(0,\sigma^{2})$ elements onto a subspace of dimension $s$, for $s \leq M$, is another vector of dimension $s$ with i.i.d. $\calCN(0,\sigma^{2})$ elements.
\end{lemma}
\begin{lemma}\label{prop1}
Assume $\{Y_{i}\}$ are independent Gamma RVs with parameters $\mu_{i}$ and $\theta_{i}$. The RV $W \sim \Gamma(\mu,\theta)$ has the same first and second order moments as the RV $Y = \sum_{i}X_{i}$, where
%\vspace{-5pt}
\begin{equation}\label{eq:secmm}
\mu = \frac{(\sum_{i}\mu_{i}\theta_{i})^{2}}{\sum_{i}\mu_{i}\theta_{i}^{2}}~~\text{and}~~\theta = \frac{\sum_{i}\mu_{i}\theta_{i}^{2}}{\sum_{i}\mu_{i}\theta_{i}}.
\end{equation}
\end{lemma}

\begin{lemma}\label{lemmeij}
Let $X\sim \Gamma(\mu,\theta)$, then $\bbE_{X}[\log_{2}(1+X)]$ is computed as
%\vspace{-10pt}
\begin{equation}
\bbE_{X}[\log_{2}(1+X)] = \frac{1}{\Gamma(\mu)\ln 2}G^{1,3}_{3,2}\left[ \theta\left\vert
\begin{array}{c}
1-\mu,1,1 \\
1,0 \\
\end{array}
\right. \right] \nonumber
\end{equation}
where $G^{m,n}_{p,q}\left[ x\left\vert
\begin{array}{c}
\eta_{1},\ldots,\eta_{p} \\
\nu_{1},\ldots,\nu_{q} \\
\end{array}
\right. \right]$ denotes the Meijer's G-function~\cite[Eq. (9.301)]{ieeecl:Table}.
\end{lemma}

\begin{proof}
\begin{align*}
\bbE_{X}[\log_{2}(1+X)] &\overset{(a)}{=} \frac{1}{\theta^{\mu}\Gamma(\mu)\ln 2}\int_0^{\infty} G^{1,2}_{2,2}\left[ x \left\vert
\begin{array}{c}
1,1 \\
1,0 \\
\end{array}%
\right. \right]
x^{\mu-1}e^{-x/\theta}dx \nonumber \\
& \overset{(b)}{=} \frac{1}{\Gamma(\mu)\ln 2}G^{1,3}_{3,2}\left[ \theta\left\vert
\begin{array}{c}
1-\mu,1,1 \\
1,0 \\
\end{array}%
\right. \right].
\end{align*}
Here, (a) follows by expressing the logarithmic term $\ln(1+x)$ via a Meijer's $G$-function
according to~\cite[Eq. (8.4.6.5)]{ieeecl:Prud} and (b) results by evaluating the integral expression in (a) using the integral identity for Meijer's $G$-functions from~\cite[Eq. (7.813.1)]{ieeecl:Table} and some algebraic simplifications.
%Note that the Meijer-$G$ function can be easily evaluated and efficiently programmed in most standard software packages (e.g. MAPLE, MATHEMATICA, and MATLAB R2011a).
\end{proof}

%%%%%%%%%%%%%%%%%%%%%%%%%%%%%%%%%%%%%%%%%%%%%%%%%%%%%%%%%%
%                      References
%%%%%%%%%%%%%%%%%%%%%%%%%%%%%%%%%%%%%%%%%%%%%%%%%%%%%%%%%%
\vspace{-5pt}
\small
\bibliographystyle{IEEEtran}

\end{document}